\documentclass[11pt,a4paper]{article}
\usepackage[T1]{fontenc}
\usepackage[utf8]{inputenc}
\usepackage{authblk}
\usepackage{algpseudocode}
\usepackage{url}
\usepackage{amsmath}
\usepackage{algorithm}
\usepackage{amssymb}
\usepackage{amsthm}
\usepackage{mathtools}
\usepackage{caption}
\usepackage{graphicx, subcaption}
\usepackage{algorithm}
\usepackage{enumerate}   
\usepackage{color}
\usepackage{amsfonts}
\usepackage{xspace}
\usepackage{wrapfig}
\usepackage{hyperref}
\usepackage{float}
\hypersetup{
    colorlinks=true,
    linkcolor=blue,
    filecolor=magenta,      
    urlcolor=cyan,
}
\usepackage{geometry}
\usepackage{bibnames}
\usepackage{cite}
\usepackage{enumitem}

\newtheorem{corollary}{Corollary}[section]

\newtheorem{obs}{Observation}[section]

\newtheorem{proposition}{Proposition}[section]

\definecolor{darkred}{rgb}{1, 0.1, 0.3}
\definecolor{darkgreen}{rgb}{0.15, 0.65, 0.15}
\definecolor{darkblue}{rgb}{0.1, 0.1, 1}
\definecolor{yellow}{rgb}{0.7, 0.7, 0}

\newtheorem{theorem}{Theorem}[section]

\newtheorem{definition}{Definition}[section]
\newcounter{claimcounter}
\numberwithin{claimcounter}{cl}

\newcommand{\permodule}     {{\mathbb{H}}}
\newcommand{\weightf}      {{w}} 

\newcommand{\csize}{\min\{ \arbor(G)m, \sum_{(u,v)\in E} (\din(u) + \dout(v) )\}}
\newcommand{\denselist}{\itemsep 0pt\parsep=1pt\partopsep 0pt}
\newcommand{\HH}    {{\mathsf H}}
\newcommand{\candidateC}    {{\mathsf{C}}} 
\newcommand{\mytriple}  {{triple-list\xspace}}
\newcommand{\myquad}    {{quadruple-list\xspace}}
\newcommand{\atriple}    {{\xi}}
\newcommand{\oldG}          {{G^{(s-1)}}}
\newcommand{\newG}          {{G^{(s)}}}

\newcommand{\arbor}         {{\mathsf{a}}}
\newcommand{\din}           {{d_{in}}}
\newcommand{\dout}          {{d_{out}}}
\newcommand{\myQ}           {{R}}

\begin{document}
\title{\Large\bfseries An efficient algorithm for $1$-dimensional (persistent) path homology}
\author[1]{Tamal K. Dey \thanks{tamaldey@cse.ohio-state.edu}}
\author[1]{Tianqi Li \thanks{li.6108@osu.edu}}
\author[1]{Yusu Wang\thanks{yusu@cse.ohio-state.edu}}
\affil[1]{
Department of Computer Science and Engineering, The Ohio State University, Columbus}
\maketitle

\begin{abstract}
This paper focuses on developing an efficient algorithm for analyzing a directed network (graph)
from a topological viewpoint. A prevalent technique for such topological analysis involves computation of homology groups and their persistence. These concepts are well suited for spaces that are not directed. As a result, one needs a concept of homology that accommodates orientations in input space.
Path-homology developed for directed graphs by~\cite{grigor2012homologies} has been effectively adapted for this purpose recently by Chowdhury and M{\'e}moli~\cite{chowdhury2018persistent}. They also give an algorithm to compute this path-homology.
Our main contribution in this paper is an algorithm that computes this path-homology and its persistence more efficiently for the $1$-dimensional ($H_1$) case.
In developing such an algorithm, we discover various structures and their efficient computations that aid computing the $1$-dimensional path-homnology. We implement our algorithm and present some preliminary experimental results.
\end{abstract}
\section{Introduction}
When it comes to graphs, traditional topological data analysis has focused mostly on undirected ones. However, applications in social networks~\cite{milo2002network, chen2009efficient}, brain networks~\cite{varshney2011structural}, and others require processing directed graphs.
Consequently, topological data analysis for these applications needs to be adapted accordingly to account for directedness. Recently, some work~\cite{chowdhury2018functorial,masulli2016topology} have
initiated to address this important but so far neglected issue. 

Since topological data analysis uses persistent homology as a main tool, one needs a notion of homology for directed graphs. Of course, one can forget the directedness and consider the underlying undirected graph as a simplicial $1$-complex and use a standard persistent homology pipeline for the analysis. However, this is less than desirable because the important information involving
directions is lost. Currently, there are two main approaches that have been proposed for dealing with
directed graphs. One uses directed clique complexes~\cite{masulli2016topology, dotko2016topological} and the other uses the concept of path homology~\cite{grigor2012homologies}. In the first approach, a $k$-clique in the input directed graph is turned into a $(k-1)$-simplex if the clique has a single source and a single sink. The resulting simplicial complex is subsequently analyzed with the usual persistent homology pipeline. One issue with this approach is that there could be very few cliques with the required condition and thus accommodating only a very few higher dimensional simplices. In the worst case, only the undirected graph can be returned as the directed clique complex if each 3-clique is a directed cycle. The second approach based on path homology alleviates this deficiency. Furthermore, certain natural functorial properties, such as K\"{u}nneth formula, do not hold for the clique complex ~\cite{grigor2012homologies}.

The path homology, originally proposed by Grigoryan, Lin, Muranov and Yau in 2012~\cite{grigor2012homologies} and later studied by ~\cite{grigor2014homotopy,grigor2015cohomology,chowdhury2018persistent}, has several properties that make it a richer mathematical structure. For example, there is a concept of homotopy under which the path homology is preserved; it accommodates K\"unneth formula; and the path homology theory is dual to the cohomology theory of digraphs introduced in~\cite{grigor2015cohomology}.
Furthermore, persistent path homology developed in \cite{chowdhury2018persistent} is shown to respect a stability property for its persistent diagrams.

To use path homologies effectively in practice, one needs efficient algorithms to compute them. 
In particular, we are interested in developing efficient algorithms for computing $1$-dimensional  path homology and its persistent version because even for this case the 
current state of the art is far from satisfactory: Given a directed graph $G$ with $n$ vertices, the most efficient algorithm proposed in~\cite{chowdhury2018persistent} has a time complexity $O(n^9)$ (more precisely, their algorithm takes $O(n^{3+3d})$ to compute the ($d-1$)-dimensional persistent path-homology).

The main contribution of this paper is stated in Theorem~\ref{thm:main}. 
The reduced time complexity of our algorithm can be attributed to the fact that we compute the
boundary groups more efficiently. In particular, it turns out that for $1$-dimensional path homology, the boundary group is determined by bigons, certain triangles, and certain quadrangles in the input directed graph. The bigons and triangles can be determined relatively easily. It is the boundary quadrangles whose computation and size determine the time complexity. The authors in~\cite{chowdhury2018persistent} compute a basis of these boundary quadrangles by constructing a certain generating set for the 2-dimensional chain group by a nice column reduction algorithm (being different from the standard simplicial homology, it is non-trivial to do reduction for path homology). 
We take advantage of the concept of \emph{arboricity} and related results in graph theory, together with other efficient strategies,
to enumerate a much smaller set generating the boundary quadrangles. Computing the cycle and boundary groups efficiently both for non-persistent and persistent homology groups is the key to our improved time complexity.

\begin{theorem}\label{thm:main}
Given a directed graph $G$ with $n$ vertices and $m$ edges, set $r = \min\{ \arbor(G)m,$ $ \sum_{(u,v)\in E} (\din(u)$ $ + \dout(v))\}$, where $\arbor(G)=O(n)$ is the so-called arboricity of $G$, and $d_{in}(u)$ and $d_{out}(u)$ are the in-degree and out-degree of $u$, respectively. 
There is an $O(r m^{\omega-1}+m\alpha (n))$ time algorithm for computing the \emph{$1$-dimensional persistent path homology} for $G$ where $\omega < 2.373$ is the exponent for matrix multiplication\footnote{That is, the fastest algorithm to multiply two $r \times r$ matrices takes time $O(r^\omega)$.}, and $\alpha (\cdot)$ is the inverse Ackermann function.

This also gives an $O(r m^{\omega-1}+m\alpha (n))$ time algorithm for computing the $1$-dimensional path homology $\HH_1$ of $G$.

In particular, for a planar graph $G$, $\arbor(G) = O(1)$ and the time complexity becomes $O(n^\omega)$. 
\end{theorem}

The \emph{arboricity} $\arbor(G)$ of a graph $G$ mentioned in Theorem~\ref{thm:main} denotes the minimum number of edge-disjoint spanning forests into which G can be decomposed~\cite{harary1971graph}. 
It is known that in general, $\arbor(G) = O(n)$, but it can be much smaller. For example, $\arbor(G)=O(1)$ for planar graphs and $\arbor(G) = O(g)$ for a graph embedded on a genus-$g$ surface \cite{chiba1985arboricity}. 
Hence, for planar graphs, we can compute $1$-dimensional persistent path homology in $O(n^{\omega})$ time whereas the algorithm in~\cite{chowdhury2018persistent} takes $O(n^{5})$ time \footnote{The original time complexity stated in the paper is $O(n^9)$ for $1$-dimensional case. However, one can improve it to $O(n^5)$ by a more refined analysis for planar graphs.}.

\paragraph{Organization of the paper.} 

After characterizing the $1$-dimensional path homology group $\mathsf{H}_1$ in Section \ref{sec:characterization}, we first propose a simple algorithm to compute it. In Section~\ref{sec:perH1}, we consider its persistent version and present an improved and more efficient algorithm. In Section~\ref{sec:application}, we develop an algorithm to compute $1$-dimensional \emph{minimal path homology basis}~\cite{dey2018efficient, erickson2005greedy}, and also show  experiments demonstrating the efficiency of our new algorithms.

\section{Background}\label{sec:bg}

We briefly introduce some necessary background for path homology. Interested readers can refer to \cite{grigor2012homologies} for more details. The original definition can be applied to structures beyond directed graphs; but for simplicity, we use directed graphs to introduce the notations. 

Given a directed graph $G = (V, E)$, we denote $(u, v)$ as the directed edge from $u$ to $v$. A \emph{self-loop} is defined to be the edge $(u, u)$ from $u$ to itself. Throughout this paper, we assume that $G$ does not have self-loops. We also assume that $G$ does not have multi-edges, i.e. for every ordered pair $u, v$, there is at most one directed edge from $u$ to $v$. 
For notational simplicity, we sometimes use index $i$ to refer to vertex $v_i \in V = \{v_1, \ldots, v_n\}$. 

Let $\mathbb{F}$ be a field with 0 and 1 being the additive and  multiplicative identities respectively. We use $-a$ to denote the additive inverse of $a$ in $\mathbb{F}$.
An {\em{elementary $d$-path}} on $V$ is simply a sequence $i_0, i_1, \cdots, i_d$ of $d+1$ vertices in $V$. We denote this path by $e_{i_0, i_1, \cdots, i_d}$. 
Let $\Lambda_d = \Lambda_d (G, \mathbb{F})$ denote the $\mathbb{F}$-linear space of all linear combinations of elementary $d$-paths with coefficients from $\mathbb{F}$. It is easy to check that the set $\{e_{i_0,\cdots, i_d} \mid i_0, \cdots, i_d \in V \}$ is a basis for $\Lambda_d$. 
Each element $p$ of $\Lambda_d$ is called a \emph{$d$-path}, and it can be written as 
\[p =\sum\nolimits_{i_0, \cdots, i_d \in V}a_{i_0\cdots i_d}e_{i_0\cdots i_d}, \text{ where }a_{i_0\cdots i_d} \in \mathbb{F}. \]

Similar to simplicial complexes, there is a well-defined \emph{boundary operator $\partial :\Lambda_d\to \Lambda_{d-1}$}:  
\[\partial e_{i_0\cdots i_d}=\sum_{i_0, \cdots, i_d \in V}(-1)^je_{i_0\cdots\hat{i}_j\cdots i_d},\] 
where $\hat{i}_k$ means the omission of index $i_k$.
The boundary of a path $p = \sum\nolimits_{i_0, \cdots, i_d \in V} a_{i_0\cdots i_d} \cdot e_{i_0\cdots i_d}$, is thus $\partial p = \sum\nolimits_{i_0, \cdots, i_d \in V} a_{i_0\cdots i_d} \cdot \partial e_{i_0\cdots i_d}$.
We set $\Lambda_{-1}=0$ and note that $\Lambda_0$ is the set of $\mathbb{F}$-linear combinations of vertices in $V$. 

Lemma 2.4 in \cite{grigor2012homologies} shows that $\partial^2=0$. 

Next, we restrict to real paths in directed graphs. Specifically, given a directed graph $G = (V, E)$, call an elementary $d$-path $e_{i_0, \cdots, i_d}$ \emph{allowed} if there is an edge from $i_k$ to $i_{k+1}$ for all $k$. Define $\mathcal{A}_d$ as the space of all allowed $d$-paths, that is, $\mathcal{A}_d:={\rm span}\{e_{i_0\cdots i_d} : e_{i_0\cdots i_d} \text{ is allowed}\}$.
An elementary $d$-path $i_0\cdots i_d$ is called \emph{regular} if $i_k \neq i_{k+1}$ for all $k$, and is \emph{irregular} otherwise. Clearly, every allowed path is regular since there is no self-loop. However, the boundary map $\partial$ on $\Lambda_d$ may create a term resulting into an irregular path.
For example, $\partial e_{uvu}=e_{vu}-e_{uu}+e_{uv}$ is irregular because of the term $e_{uu}$. 
To deal with this case, the term containing consecutive repeated vertices is identified with $0$~\cite{grigor2012homologies}. 
Thus, for the previous example, we get $\partial e_{uvu}=e_{vu}-0+e_{uv}=e_{vu}+e_{uv}$. The boundary map $\partial$ on $\mathcal{A}_d$ is taken to be the boundary map for $\Lambda_d$ restricted on $\mathcal{A}_d$ with this modification:
where all terms with consecutive repeated vertices created by the boundary map $\partial$ are replaced with $0$'s.

Unfortunately, after restricting to the space of allowed paths $\mathcal{A}_*$, the inclusion that $\partial \mathcal{A}_d \subset \mathcal{A}_{d-1}$ may not hold any more; that is, the boundary of an allowed $d$-path is not necessarily an allowed $(d-1)$-path. 
To this end, we adopt a stronger notion of allowed paths: an allowed path $p$ is \emph{$\partial$-invariant} if $\partial p$ is also allowed. 
Let $\Omega_d := \{p \in \mathcal{A}_d \mid \partial p \in \mathcal{A}_{d-1}\}$ be the space generated by all $\partial$-invariant paths. We then have 
$\partial \Omega_d \subset \Omega_{d-1}$ (as $\partial^2 = 0$). 
This gives rise to the following \emph{chain complex of $\partial$-invariant allowed paths}: 
\[\cdots \Omega_d\xrightarrow[]{\partial}\Omega_{d-1}\xrightarrow[]{\partial}\cdots\Omega_d\xrightarrow[]{\partial}\Omega_0\xrightarrow[]{\partial}0. \]

We can now define the homology groups of this chain complex.
The $d$-th cycle group is defined as $\mathsf{Z}_d={\rm Ker}\, \partial|_{\Omega_d}$, and elements in $\mathsf{Z}_d$ are called $d$-cycles. The $d$-th boundary group is defined as $\mathsf{B}_d={\rm Im}\,\partial|_{\Omega_{d+1}}$, with elements of $\mathsf{B}_d$ being called $d$-boundary cycles (or simply $d$-boundaries). 
 The resulting \emph{$d$-dimensional path homology group} is
$\mathsf{H}_d(G, \mathbb{F}) = \mathsf{Z}_d /\mathsf{B}_d$.

\begin{figure}[H]
   \centering
    \begin{subfigure}[t]{0.3\textwidth}
        \centering
        \includegraphics[width=3cm]{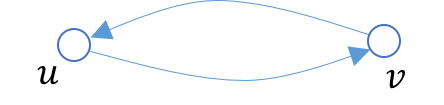}
        \subcaption{Bigon}
    \end{subfigure}%
    ~ 
    \begin{subfigure}[t]{0.3\textwidth}
        \centering
        \includegraphics[width=3cm]{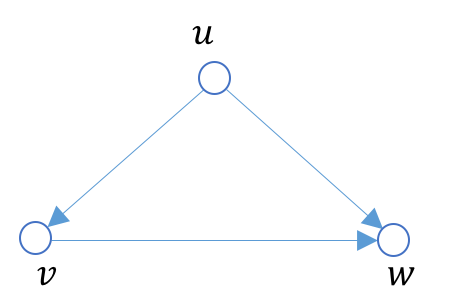}
        \subcaption{ Boundary triangle}
    \end{subfigure}
    ~
    \begin{subfigure}[t]{0.3\textwidth}
        \centering
        \includegraphics[width=3cm]{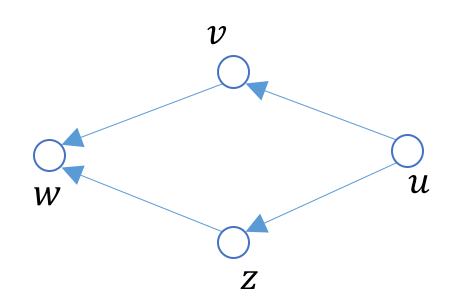}
        \subcaption{ Boundary quadrangle}
    \end{subfigure}
    \caption{Examples of 1-boundaries}
    \label{fig:1bnd}
\end{figure}

\subsection{Examples of 1-boundaries}\label{subsec:bgexamples} 

Below we give three examples of 1-boundaries; see Figure~\ref{fig:1bnd}.
\paragraph{Bi-gon:}A \emph{bi-gon} is a 1-cycle $e_{uv} + e_{vu}$ consisting of two edges $(u,v)$ and $(v, u)$ from $E$; see Figure~\ref{fig:1bnd}(a).  
Consider the 2-path $e_{uvu}$. We have that its boundary is $\partial(e_{uvu})=e_{vu}-e_{uu}+e_{uv} = e_{vu} + e_{uv}$. Since both $e_{vu}$ and $e_{uv}$ are allowed 1-paths, it follows that any bi-gon $e_{vu}+e_{uv}$ of $G$ is necessarily a 1-boundary.

\paragraph{Boundary triangle:}
Consider the 1-cycle $C=e_{vw}-e_{uw}+e_{uv}$ of $G$ (it is easy to check that $\partial\, C=0$). 
Now consider the 2-path $e_{uvw}$: its boundary is then $\partial(e_{uvw})=e_{vw}-e_{uw}+e_{uv}=C$. 
Note that every summand in the boundary is allowed. Thus $C$ is a 1-boundary. 
We call any triangle in $G$ isomorphic to $C$ a \emph{boundary triangle}. 
Note that a boundary triangle always has one sink and one source; see the source $u$ and sink $w$ in Figure~\ref{fig:1bnd}(b). In what follows, we use $(u, w \mid v)$ to denote a boundary triangle where $u$ is the source and $w$ is the sink.

\paragraph{Boundary quadrangle:}
Consider the 1-cycle $C=e_{uv}+e_{vw}-e_{uz}-e_{zw}$ from $G$. It is easy to check that $C$ is the boundary of the $2$-path $e_{uvw}-e_{uzw}$, as $\partial(e_{uvw}-e_{uzw})=e_{vw}-e_{uw}+e_{uv}-(e_{zw}-e_{uw}+e_{uz})=e_{vw}+e_{uv}-e_{zw}-e_{uz}=C$.
We call any quadrangle isomorphic to $C$ a \emph{boundary quadrangle}.

In the remainder of the paper, we use $\myQ(u, v, w,z)$ to represent a quadrangle; i.e, a 1-cycle consisting of 4 edges whose \emph{undirected version} has the form $(u,v) + (v, w) + (w,z) + (z,u)$.  (Note that a quadrangle may not be a boundary quadrangle). 
We denote a boundary quadrangle $e_{uv}+e_{vw}-e_{uz}-e_{zw}$ by $\{u, w \mid v, z\}$, where $u$ and $w$ are the source and sink of this boundary quadrangle respectively.

\section{Computing \texorpdfstring{$1$}{Lg}-dimensional path homology \texorpdfstring{$\HH_1$}{Lg}}\label{sec:characterization}

Note that the $1$-dimensional $\partial$-invariant path space $\Omega_1=\Omega_1(G)$ is the space generated by all edges~\cite{grigor2012homologies} because the boundary of every edge is allowed by definition. 

Now consider the $1$-cycle group $\mathsf{Z}_1\subseteq \Omega_1$; that is,
$\mathsf{Z}_1$ is the kernel of $\partial$ applied to $\Omega_1$.
We show below that a basis of $\mathsf{Z}_1$ can be computed by considering a spanning tree of 
the undirected version of $G$, which denoted by $G_u$. This is well known when $\mathbb{F}$ is $\mathbb{Z}_2$. 
It is easy to see that this spanning tree based construction also works for arbitrary field $\mathbb{F}$. 

Specifically, let $T$ be a rooted spanning tree of $G_u$ with root $r$, and $\bar{T}:= G_u\setminus T$.
For every edge $e=(v_1, v_2)\in \bar T$, let $c_e$ be the
$1$-cycle (under $\mathbb{Z}_2$) obtained by summing $e$ and all edges on the paths $\pi_1$ and $\pi_2$ between $v_1$ and $r$, and $v_2$ and $r$ respectively. The cycles $\{c_e, e\in \bar T\}$ form a basis of $1$-cycle group of $G_u$ under $\mathbb{Z}_2$ coefficient. Now for every such cycle $c_e$ in $G_u$, we also have a cycle in $\Omega_1(G)$ containing same edges with $c_e$ which are assigned
a coefficient $1$ or $-1$ depending on their orientations in $G$. We call this $1$-cycle in $\Omega_1(G)$ also $c_e$. Then we have the following proposition.

\begin{proposition}\label{prop:Z1}
The cycles $\{c_e|e\in \bar T\}$ in $\Omega_1(G)$ form a basis for $\mathsf{Z}_1$ under any coefficient field $\mathbb{F}$.
\end{proposition}
\begin{proof}
First, it is obvious that the cycles $\{c_e|e\in \bar T\}$ are independent because every cycle contains a unique edge that does not exist in other cycles, which means any cycle cannot written as linear combination of other cycles in the set. Now what remains is to show that the cycles $\{c_e|e\in \bar T\}$ generate $\mathsf{Z}_1$. Consider any $1$-cycle $c\not=0$ in $\mathsf{Z}_1$ with coefficients in $\mathbb{F}$. Observe that $c$ must have at least one edge from $\bar{T}$ because otherwise it will have non-zero coefficients only on edges in $T$ whose boundary cannot be $0$. Consider an edge $e_0\in \bar{T}$ with non-zero coefficient $a_0\in \mathbb{F}$ in $c$. The cycle $c'= -a_0\cdot c_{e_0} + c$ has $e_0$ with zero coefficient.
If $c'$ is not $0$, continue the argument again and we are guaranteed to derive a null cycle ultimately because every time we make the coefficient of an edge belonging to $\bar T$ zero.
This means that we have $c-a_0c_{e_0}-\cdots-a_kc_{e_k}=0$ for some $k$. In other words
$c=\sum_{i=0}^k a_i c_{e_i}$. It immediately follows that $\{c_e|e\in \bar{T}\}$ form a basis for
$\mathsf{Z}_1$.
\end{proof}

Now, we show a relation between $1$-dimensional homology, cycles, bigons, 
triangles and quadrangles. 
Recall that bi-gons, boundary triangles and boundary quadrangles are specific types of $1$-dimensional boundaries with two, three or four vertices, respectively; see Section \ref{subsec:bgexamples}. 
The following theorem is similar to Proposition 2.9 from \cite{grigor2014homotopy}, where the statement there is under coefficient ring $\mathbb{Z}$. For completeness, we include the (rather similar) proof for our case in Appendix~\ref{apex:bnd}. 

\begin{theorem}\label{thm:B1}
Let $G=(V, E)$ be a directed graph. 
Let $\mathsf{Q}$ denote the space 
generated by all boundary triangles, boundary quadrangles and bi-gons in $G$. 
Then we have $\mathsf{B}_1= \mathsf{Q}$.
\end{theorem}

\begin{corollary}\label{cor:H1characterization}
The $1$-dimensional path homology group satisfies that $\mathsf{H}_1 = \mathsf{Z}_1/\mathsf{Q}$.
\end{corollary}

\subsection{A simple algorithm}\label{sec:slowalg}

Theorem \ref{thm:B1} and Corollary \ref{cor:H1characterization} provide us a simple framework to compute $\HH_1$. Below we only focus on the computation of the rank of $\HH_1$; but the algorithm can easily be modified to output a basis for $\HH_1$ as well. 
Later in Section \ref{sec:perH1}, we will develop a 
more efficient and sophisticated algorithm for the $1$-dimensional \emph{persistent} path homology $\HH_1$, which as a by-product, also gives a more efficient algorithm to compute $\HH_1$. 

In the remaining of this paper, we represent each cycle in $\mathsf{Z}_1$ with a vector.
Assume all edges are indexed from $1$ to $m$ as $e_1, \cdots, e_m$ where $m$ is the number of edges. Then, each 1-cycle $C$ is an $m$-dimensional vector, where $C[i] \in \mathbb{F}$ records the coefficient for edge $e_i$ in $C$. 

\begin{algorithm}[H]
\caption{A simple first algorithm to compute rank of $\HH_1$}
\label{alg:slowAlg}
 \begin{algorithmic}[1] 
 	\Procedure{CompH1-Simple}{$G$, $t$}
	\State {\sf (Step 1)}: Compute rank of 1-cycle group $\mathsf{Z}_1$
	\State {\sf (Step 2)}: Compute rank of 1-boundary group $\mathsf{B}_1$ 
	\State ~~~~~~{\sf (Step 2.a)} Compute a \emph{generating set} $\candidateC$ of 1-boundary cycles that generates $\mathsf{B}_1$
	\State ~~~~~~{\sf (Step 2.b)}
	From $\candidateC$ compute a basis for $\mathsf{B}_1$
	\State Return $rank(\HH_1) = rank(\mathsf{Z}_1) - rank(\mathsf{B}_1)$. 
	\EndProcedure
\end{algorithmic}
\end{algorithm}

\subparagraph*{{\sf (Step 1)}: cycle group $\mathsf{Z}_1$.}By Proposition \ref{prop:Z1}, $rank(\mathsf{Z}_1) = |E| - |V| + 1$ for directed graph $G = (V, E)$. The computation of the rank takes $O(1)$ time (or $O(|V|^2)$ time if we need to output a basis of it explicitly).  

\subparagraph*{{\sf (Step 2)}: boundary group $\mathsf{B}_1$.} 
Note that by Theorem \ref{thm:B1}, we can compute the set of all bigons, boundary triangles and boundary quadrangles as a \emph{generating set} $\candidateC$ of 1-boundary cycles (meaning that it generates the boundary group $\mathsf{B}_1$) for {\sf (Step 2.a)}. 
However, such a set $\candidateC$ could have size $\Omega(n^2)$ even for a planar graph, where $n = |V|$; see Figure \ref{fig:quadratic}. (For a general graph, the number of boundary quadrangles could be $\Theta(n^4)$.) 

\begin{figure}[H]
\includegraphics[width=5cm]{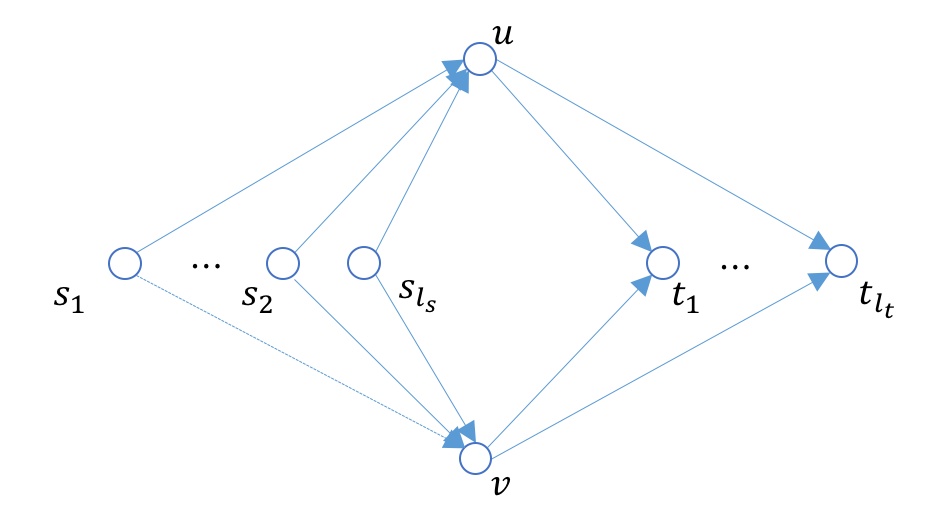}
\centering
\caption{There are $n$ vertices but $l_s\cdot l_t=\Theta(n^2)$ quadrangles, $l_s=\lfloor(n-2)/2\rfloor$ and $l_t=\lceil(n-2)/2\rceil$}
\label{fig:quadratic}
\end{figure}

To make {\sf (Step 2.b)} efficient, we wish to have a generating set $\candidateC$ of 1-boundary cycles with \emph{small cardinality}. To this end, we leverage a classical result of \cite{chiba1985arboricity} to reduce the size of $\candidateC$.

Given an undirected graph $G$ where the number of multi-edges between any two vertices is constant, its \emph{arboricity $\arbor(G)$} 
is the minimum
number of edge-disjoint spanning forests which G can be decomposed into~\cite{harary1971graph}. 
An alternative definition is 
\[\arbor(G) = \max_{H ~\text{is~a~subgraph~of}~ G} \frac{|E(H)|}{|V(H)|-1}.\]
From this definition, it is easy to see (and well-known, see e.g, \cite{chiba1985arboricity}) that:
\begin{obs}\label{obs:arboricity}
Given an undirected graph $G$ where the number of multi-edges between any two vertices is constant:\\
(1). If $G$ is a planar graph, or a graph with bounded vertex degrees, then $\arbor(G) = O(1)$. \\
(2). If $G$ is a graph embedded on a genus $g$ surface, then $\arbor(G) = O(g)$. \\
(3). In general, if $G$ does not contain self-loops, then $\arbor(G) = O(n)$.
\end{obs}
We will leverage some classical results from \cite{chiba1985arboricity}.  
First, to represent quadrangles, we use the following \emph{\mytriple} representation \cite{harary1971graph} 
: a \mytriple{} $(u, v, \{w_1, w_2,\cdots, w_l\})$ means that for each $i$, $w_i$ is adjacent to both $u$ and $v$, where we say $u'$ and $v'$ are adjacent if either $(u',v')$ or $(v', u')$ are in $E$ (i.e, $u'$ and $v'$ are adjacent when disregarding directions). 
Given such a \mytriple{} $\atriple = (u, v, \{w_1, w_2,\cdots, w_l\})$, it is easy to see that $u, w_i, v, w_j$ form the \emph{consecutive vertices} of a quadrangle in the undirected version of graph $G$; and we also say that the undirected quadrangle $R(u, w_i, v, w_j)$ is \emph{covered} by this \mytriple{}. We say that a vertex $z$ is in a \mytriple{} $(u, v, \{w_1, w_2,\cdots, w_l\})$ if it is in the set $\{w_1, w_2,\cdots, w_l\}$.

The size of a \mytriple{} is the total number of vertices contained in it. This \mytriple{} $\atriple$ thus represents $\Theta(l^2)$ number of undirected quadrangles in $G$ succinctly with $\Theta(l)$ size.  

\begin{proposition}[\cite{chiba1985arboricity}]\label{prop:enumeration}
(1) Let $G$ be a connected undirected graph with $n$ vertices and $m$ edges. There is an algorithm listing all the triangles in $G$ in $O(\arbor(G)m)$ time. \\
(2) There is an algorithm to compute a set of \mytriple{}s which covers 
all quadrangles in a
connected graph G in $O(\arbor(G)m)$ time. The total size complexity of all \mytriple{}s is $O(\arbor(G) m)$. 
\end{proposition}

Using the above result, we can have the following theorem.
\begin{theorem}\label{thm:candidate1}
Let $G=(V, E)$ be a directed graph with $n$ vertices and $m$ edges. 
We can compute a generating set $\candidateC$ of 1-boundary cycles for $\mathsf{B}_1$ with cardinality $O(\arbor(G)m)$ in time $O(\arbor(G)m)$. 
\end{theorem}
\begin{proof}
We will show that we can generate all bigons, 
boundary triangles and boundary quadrangles, via a generating set $\candidateC$ satisfying the requirement in the theorem. 

First, note that the total number of bigons are $O(m)$ since there are $m$ directed edges and for each edge, there is at most one bi-gon created as there is no multi-edge in $G$. 
Next, by Proposition~\ref{prop:enumeration}, we have that the total number of undirected triangles is $O(\arbor(G)m)$, which further bounds the total number of directed triangles, as well as that of the boundary triangles by $O(\arbor(G)m)$. In particular, we enumerate all $O(\arbor(G)m)$ undirected triangles in $G$ in $O(\arbor(G)m)$ time, and if the vertices form a boundary triangle, we add it to the generating set $\candidateC$. 
This adds all $O(\arbor(G)m)$ number of boundary triangles to $\candidateC$ in $O(\arbor(G)m)$ time. 

The case for quadrangles is more involved. The total number of quadrangles could be $\Theta(n^4)$
and we aim to find a subset of $O(\arbor(G)m)$ that generate all boundary quadrangles to add to $\candidateC$. 

By Proposition \ref{prop:enumeration}, in $O(\arbor(G)m)$ time, we can compute a list $L$ of \mytriple{}s with $O(\arbor(G)m)$ total size complexity that generate all \emph{undirected quadrangles}. 
Note that this also implies that $|L| = O(\arbor(G)m)$. 

Now for a \mytriple{} $\atriple_{uv} = (u, v, \{w_1, w_2,\cdots, w_l\})$, 
we will generate three lists as follows (see Figure \ref{fig:cases_quad}): 

\begin{figure}[H]
\includegraphics[width=10cm]{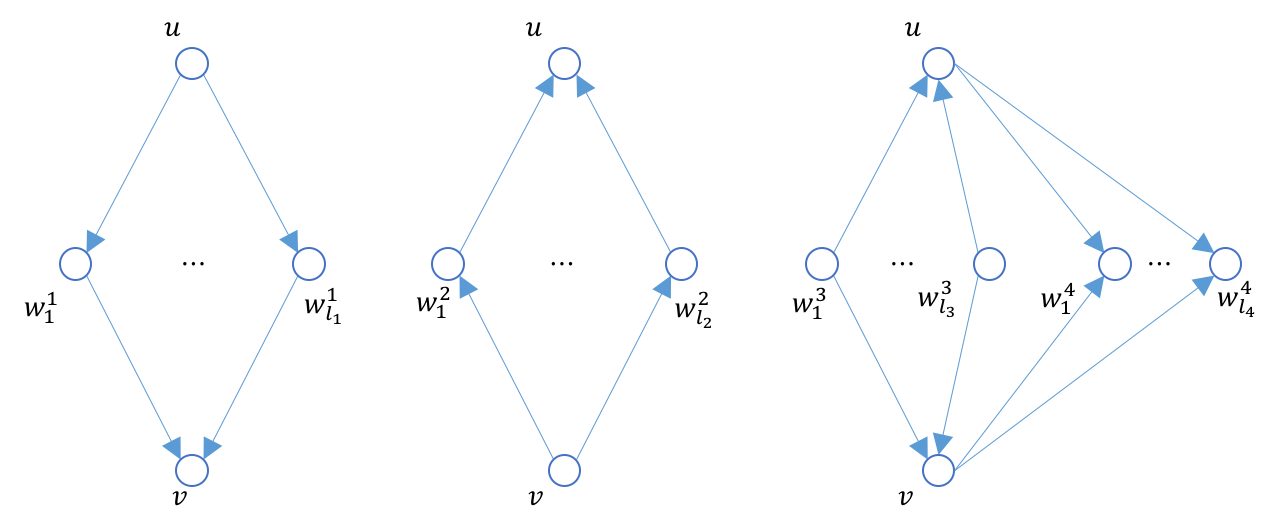}
\centering
\caption{Three lists extracted from $\atriple = (u, v, \{w_1, w_2,\cdots, w_l\})$}
\label{fig:cases_quad}
\end{figure}

\begin{description}
\item[Type-1 list:] a \mytriple{} $\atriple^{(1)}_{uv} = (u, v, \{w^1_1, w^1_2,\cdots, w^1_{l_1}\})$ where $w^1_i$ is in $\atriple^{(1)}_{uv}$ if and only if $w^1_i \in \atriple_{uv}$, and both edges $(u, w^1_i)$ and $(w^1_i, v)$ are in $E$. 
We say that a \emph{boundary} quadrangle can be \emph{generated by} the type-1 list $\atriple^{(1)}_{uv}$ if it is of the form $(u, v \mid w_i^1, w_j^1)$, with $i\neq j \in [1, l_1]$. 
\item[Type-2 list:] a \mytriple{} $\atriple^{(2)}_{uv} = (u, v, \{w^2_1, w^2_2,\cdots, w^2_{l_2}\})$ where $w^2_i$ is in $\atriple^{(2)}_{uv}$ if any only if $w^2_i$ is in $\atriple_{uv}$, and the edges $(v, w^2_i)$ and $(w^2_i, u)$ are in $E$. 
We say that a \emph{boundary} quadrangle can be \emph{generated by} the type-2 list $\atriple^{(2)}_{uv}$ if it is of the form $(v, u \mid w_i^2, w_j^2)$, with $i\neq j \in [1, l_2]$. 
\item[Type-3 list:] a so-called \emph{\myquad{}} $\atriple^{(3)}_{uv} = (u, v, \{w^3_1, w^3_2,\cdots, w^3_{l_3}\}, \{w^4_1, w^4_2,\cdots, w^4_{l_4}\})$, 
where each $w^3_i$ and $w^4_j$ are in $\atriple_{uv}$, and edges $(w^3_i, u)$, $(w^3_i, v)$, $(u, w^4_j)$ and $(v, w^4_j)$ are in $E$, for $1\le i\le l_3,1\le j\le l_4$. 
We say that a \emph{boundary} quadrangle can be \emph{generated by} the type-3 list $\atriple^{(3)}_{uv}$ if it is of the form $(w_i^3, w_j^4 \mid u, v)$, with $i \in [1, l_3]$ and $j\in [1, l_4]$. 
\end{description}

It is easy to see that $l_1+l_2+l_3+l_4=O(l)$. Let $\widehat L$ denote all lists generated by triples in $L$. Note that total size complexity of $\widehat L$ (which is the sum of the size of each \mytriple{} or \myquad{} in $\widehat L$) is still $O(\arbor(G)m)$.

Denote $\mathcal{R}$ as the set of all boundary quadrangles in $G$; note that by definition, each element in $\mathcal{R}$ is a 1-boundary. 
On the other hand, each boundary quadrangle in $\mathcal{R}$ is generated by some list in $\widehat L$. 
Indeed, given an arbitrary boundary quadrangle $(x,y \mid z, s)$, By Proposition \ref{prop:enumeration}, we know that its undirected version $R(x, z, y, s)$ must be covered by some \mytriple{} $\atriple_{uv} \in L$. If $x = u$ and $y = v$, then the boundary quadrangle is then generated by the type-1 list $\atriple^{(1)}_{uv}$. 
If $x=v$ and $y=u$, then it is generated by the type=2 list $\atriple^{(2)}_{uv}$. 
Otherwise, by the definition of \emph{$\atriple_{uv}$ covering the undirected quadrangle $R(x,z,y,s)$}, it must mean that $u = z, v = s$ or $u=s, v=z$. Thus this boundary quadrangle is generated by the type-3 list $\atriple^{(3)}_{uv}$. 
This proves that all boundary quadrangles in $\mathcal{R}$ are generated by the lists in the set $\widehat L$. 

Finally, we will add boundary quadrangles to $\candidateC$ as follows: 
We inspect each list in $\widehat L$. 
\begin{description}\denselist 
\item[Case 1]: 
If it is a type-1 list, say of the form $\atriple^{(1)}_{uv} = (u, v, \{w^1_1, w^1_2,\cdots, w^1_{l_1}\})$, we add $l_1-1$ number of boundary quadrangles of the form $(u, v \mid w_1^1, w_i^1)$ to $\candidateC$ for each $i\in [2, l_1]$. 
\item[Case 2]: If it is a type-2 list of the form $\atriple^{(2)}_{uv} = (u, v, \{w^2_1, w^2_2,\cdots, w^2_{l_2}\})$, we then add $l_2-1$ number of boundary quadrangles of the form $(v, u \mid w_1^2, w_i^2)$ to $\candidateC$, for each $i\in [2, l_2]$. 
\item[Case 3]: If it is a type-3 list of the form $\atriple^{(3)}_{uv} = (u, v, \{w^3_1, w^3_2,\cdots, w^3_{l_3}\}, \{w^4_1, w^4_2,\cdots, w^4_{l_4}\})$, then we add $l_3 + l_4$ number of boundary quadrangles of the form $(w_1^3, w_i^4\mid u, v)$ or $(w_j^3, w_1^4 \mid u, v)$ into $\candidateC$, for each $i\in [1, l_4]$ and $j \in [1, l_3]$. 
\end{description}
Note that for each case above, the subset of quadrangles we add to $\candidateC$ can generate all the boundary quadrangles generated by the corresponding list $\atriple^{(k)}_{uv}$ for $k = 1, 2,$ or $3$.  
Indeed, for (Case 1): each boundary quadrangle $(u, v \mid w_i^1, w_j^1)$ generated by $\atriple^{(1)}_{uv}$ can be written as the linear combination of $(u, v\mid w_1^1, w_j^1)$ and $(u, v\mid w_1^1, w_i^1)$, both of which are added to $\candidateC$. (Case 2) can be argued in a symmetric manner. 
For (Case 3), consider a boundary quadrangle $(w_j^3, w_i^4 \mid u, v)$ generated by $\atriple^{(3)}_{uv}$. It can be written as the combination of $(w_j^3, w_1^4 \mid u, v)$, $(w_1^3, w_1^4\mid u, v)$ and $(w_1^3, w_i^4 \mid u, v)$, all of which on the righthand side are added to $\candidateC$. 

Hence in summary, the set of quadrangles we add to $\candidateC$ will generate all boundary quadrangles $\mathcal{R}$. Furthermore, note that in each case, the number of quadrangles we add to $\candidateC$ is linear in the size of the list from $\widehat L$ being considered. 
Hence the total number of boundary quadrangles ever added to $\candidateC$ is bounded by the total size of $\widehat L$ which is further bounded by $O(\arbor(G)m)$. 
Putting everything together, the theorem then follows. 
\end{proof}

It then follows from Theorem \ref{thm:B1} that {\sf (Step 2.a)} can be implemented in $O(\arbor(G)m)$ time, producing a generating set of cardinality $O(\arbor(G)m)$. 
Finally, representing each boundary cycle in $\candidateC$ as a vector of dimension $m = |E|$, we can then compute the rank of cycles in $\candidateC$ in $O(|\candidateC| m^{\omega - 1}) = O(\arbor(G) m^\omega)$,
where $\omega< 2.373$ is the exponent for matrix multiplication~\cite{cheung2013fast}.

Putting everything together, we have that
\begin{theorem}\label{thm:slowAlg}
Given a directed graph $G = (V, E)$ with $n = |V|$ and $m = |E|$, Algorithm \ref{alg:slowAlg} computes the rank of the $1$-dimensional path homology group $\HH_1$ in $O(\arbor(G) m^\omega)$ time. 
The algorithm can be extended to compute a basis for $\HH_1$ with the same time complexity. 
\end{theorem}

For example, by Observation \ref{obs:arboricity}, if $G$ is a planar graph, then we can compute $\HH_1$ in $O(n^\omega)$. For a graph $G$ embedded on a genus $g$ surface, $\HH_1$ can be computed in $O(g n^\omega)$ time. 
In contrast, we note that the algorithm of \cite{chowdhury2018persistent} takes $O(n^5)$ time for planar graphs. 

\section{Computing persistent path homology \texorpdfstring{$\HH_1$}{Lg}}\label{sec:perH1}

The concept of arboricity used in the previous section does not consider edge directions. Indeed, our algorithm to compute a generating set $\candidateC$ as given in the proof of Theorem \ref{thm:candidate1} first computes a (succinct) representation of all quadrangles, whether they contribute to boundary quadrangles or not. 
On the other hand, as Figure~\ref{fig:nobndQ} illustrates, a graph $G$ can have no boundary quadrangle despite the fact that the graph is dense (with $\Theta(n^2)$ edges and thus $\arbor(G) = \Theta(n)$ arboricity). 
Another way to view this is that the example has no allowed $2$-path, and thus no $\partial$-invariant 2-paths and consequently no 1-boundary cycles. Our algorithm will be more efficient if it can also respect the number of allowed elementary $2$-paths. 

\begin{figure}[H]
\centering
\includegraphics[width=8cm]{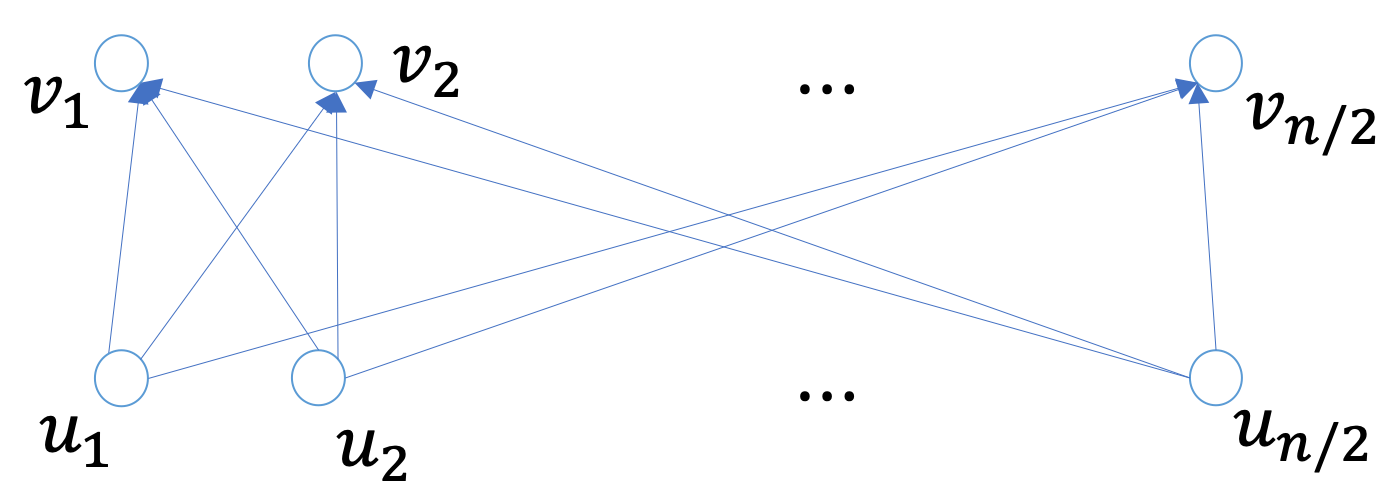}
\caption{A dense graph with no boundary quadrangle}
\label{fig:nobndQ}
\end{figure}

In fact, a more standard and natural way to compute a basis for the 1-boundary group proceeds by taking the boundary of $\partial$-invariant 2-paths. The complication is that unlike in the simplicial homology case, it is not immediately evident how to compute a basis for $\Omega_2$ (the space of $\partial$-invariant 2-paths). 
Nevertheless, Chowdhury and M\'{e}moli presented an elegant algorithm to show that a basis for $\mathsf{B}_1$ (and $\HH_1$) can still be computed using careful column-based matrix reductions~\cite{chowdhury2018persistent}. The time complexity of their algorithm is $O((\sum_{(u, v)\in E} (d_{in}(u)+d_{out}(v)))mn^2)$ which depends on the number of elementary $2$-paths \footnote{The time complexity  given in the paper~\cite{chowdhury2018persistent} assumes that the input directed graph is complete, and takes $O(n^9)$ to compute $\HH_1$. However, a more refined analysis of their time complexity shows that it can be improved to $O((\sum_{(u, v)\in E} d_{in}(u)+d_{out}(v))mn^2)$. }.

In this section, we present an algorithm that can take advantage of both of the previous approaches (the algorithm of \cite{chowdhury2018persistent} and Algorithm \ref{alg:slowAlg}). 
Similar to \cite{chowdhury2018persistent}, we will now consider the persistent path homology setting, where we will add directed edges in $G$ one by one incrementally. Hence our algorithm can compute the \emph{persistent} $\HH_1$ w.r.t. a filtration. 
However different from \cite{chowdhury2018persistent}, instead of reducing a matrix with columns corresponding to all elementary allowed 2-paths, we will follow a similar idea as in Algorithm \ref{alg:slowAlg} and add a generating set of boundary cycles each time we consider a new directed edge. 

\subsection{Persistent path homology}
We now introduce the definition of the persistent path homology~\cite{chowdhury2018persistent}. 
The \emph{persistent vector space} is a family of vector spaces together with linear maps $\{U^{\delta}\xrightarrow{\mu_{\delta, \delta'}}U^{\delta'}_{\delta \le \delta'\in \mathbb{R}}\}$ so that: (1) $\mu_{\delta, \delta}$ is the identity for every $\delta\in \mathbb{R}$; and (2) $\mu_{\delta, \delta''} = \mu_{\delta, \delta'}\circ \mu_{\delta', \delta''}$ for each $\delta\le \delta'\le \delta''\in \mathbb{R}$.

Let $G=(V, E, \weightf)$ be a weighted directed graph where $V$ is the vertex set, $E$ is the edge set, and $\weightf$ is the weight function $\weightf:E\to \mathbb{R}^+$. 
For every $\delta\in \mathbb{R}^+$, a directed graph $G^{\delta}$ can be constructed as $G^{\delta}=(V^{\delta}=V, E^{\delta}=\{e\in E: \weightf(e)\le \delta\})$. 
This gives rise to a filtration of graphs  $\{G^{\delta}\xhookrightarrow{} G^{\delta'}\}_{\delta\le\delta'\in \mathbb{R}}$ using the natural inclusion map $i_{\delta, \delta'}: G^{\delta}\xhookrightarrow{} G^{\delta'}$.
\begin{definition}\cite{chowdhury2018persistent}
The $1$-dimensional persistent path homology of a weighted directed graph $G=(V, E, \weightf)$ is defined as the persistent vector space 
$\permodule_1 :=\{\HH_1(G^{\delta})\xrightarrow{i_{\delta, \delta'}}\HH_1(G^{\delta'})\}_{\delta \le \delta'\in \mathbb{R}}.$ 
The $1$-dimensional path persistence diagram $Dg(G)$ of $G$ is the persistence diagram of $\permodule_1$.
\end{definition}

To compute the path homology $\HH_1(G)$ of an unweighted directed graph $G=(V,E)$, we can order edges in $E$ arbitrarily with the index of an edge in this order being its weight. The rank of $\HH_1(G)$ can then be retrieved from the $1$-dimensional persistent homology group induced by this filtration by considering only those homology classes that ``never die''.

\subsection{A more efficient algorithm}
\label{subsec:persistH1}
In what follows, to simplify presentation, we assume that we are given a directed graph $G = (V, E)$, where edges are already sorted $e_1, \ldots, e_m$ in increasing order of their weights. Let $G^{(i)} = (V, E^{(i)} = \{e_1, \ldots, e_i\})$ denote the subgraph of $G$ spanned by the edges $e_1, \ldots, e_i$; and set $G^{(0)} = (V, \varnothing)$. 
We now present an algorithm to compute the $1$-dimensional persistent path homology induced by the nesting sequence $G^{(0)} \subseteq G^{(1)} \subseteq \cdots G^{(m)}$. 
In particular, in Algorithm \ref{alg:fastAlg}, as we insert each new edge $e_s$, moving from $G^{(s-1)}$ to $G^{(s)}$, we maintain a basis for $\mathsf{Z}_1 (s):= \mathsf{Z}_1(G^{(s)})$ and for $\mathsf{B}_1 (s):= \mathsf{B}_1(G^{(s)})$, updated from $\mathsf{Z}_1(s-1)$ and $\mathsf{B}_1(s-1)$ and output new persistent pairs. On a high level, this algorithm follows the standard procedure in~\cite{cohen2006vines}.

\begin{algorithm}[H]
\caption{Compute $1$-D persistent path homology for a directed graph $G = (V, E)$}
\label{alg:fastAlg}
 \begin{algorithmic}[1] 
 	\Procedure{Persistence}{$G$}
 	\State Order the edges in non-decreasing order of their weights: $e_1, \ldots, e_m$. 
    \State Set $G^{(0)}=(V, \varnothing)$, current basis  for 1-boundary group is $B=\varnothing$. 
	\For{$s = 1$ to $m$}
	    \State Call \Call{GenSet}{$s$} to compute a generating set $\candidateC_s$ containing a basis for newly generated 1-boundary cycles moving from $G^{(s-1)}$ to $G^{(s)}$.
        \State Call \Call{FindPairs}{$s$} to output new persistent pairs, and update the boundary basis $B$ for $G^{(s)}$.
      \EndFor
	\EndProcedure
\end{algorithmic}
\end{algorithm}

\subsubsection{Procedure \texorpdfstring{\sc{GenSet}}{Lg}.} 
Note that $G^{(s)}$ is obtained from $G^{(s-1)}$ by inserting a new edge $e_s = (u, v)$ to $G^{(s-1)}$. 
At this point, we have already maintained a basis $B$ for $\mathsf{B}_1(\oldG)$. Our goal is to compute a set of generating boundary cycles $\candidateC_s$ such that $B \cup \candidateC_s$ contains a basis for $\mathsf{B}_1(\newG)$. 

We first inspect the effect of adding edge $e_s = (u, v)$ to $\newG$. 
Two cases can happen: 

(Case-A): The endpoints $u$ and $v$ are in different connected components in (the undirected version of) $\oldG$, and after adding $e_s$, those two components are merged into a single one in $\newG$. 
In this case, no cycle is created, nor does the boundary group change. Thus $\mathsf{Z}_1(\oldG) = \mathsf{Z}_1(\newG)$ and $\mathsf{B}_1(\oldG) = \mathsf{B}_1(\newG)$. 
We say that edge $e_s$ is \emph{negative} in this case (as it kills in $\HH_0$). 

The algorithm maintains the set of negative edges seen so far, which is known to form a spanning forest ${T}_s$ of $V$. (Here, we abuse the notation slightly and say that a set of directed edges span a tree for a set of vertices if they do so when directions are ignored.)
The algorithm maintains ${T}_s$ via a union-find data structure. 

\begin{figure}[H]
\centering
\includegraphics[width=3.5cm]{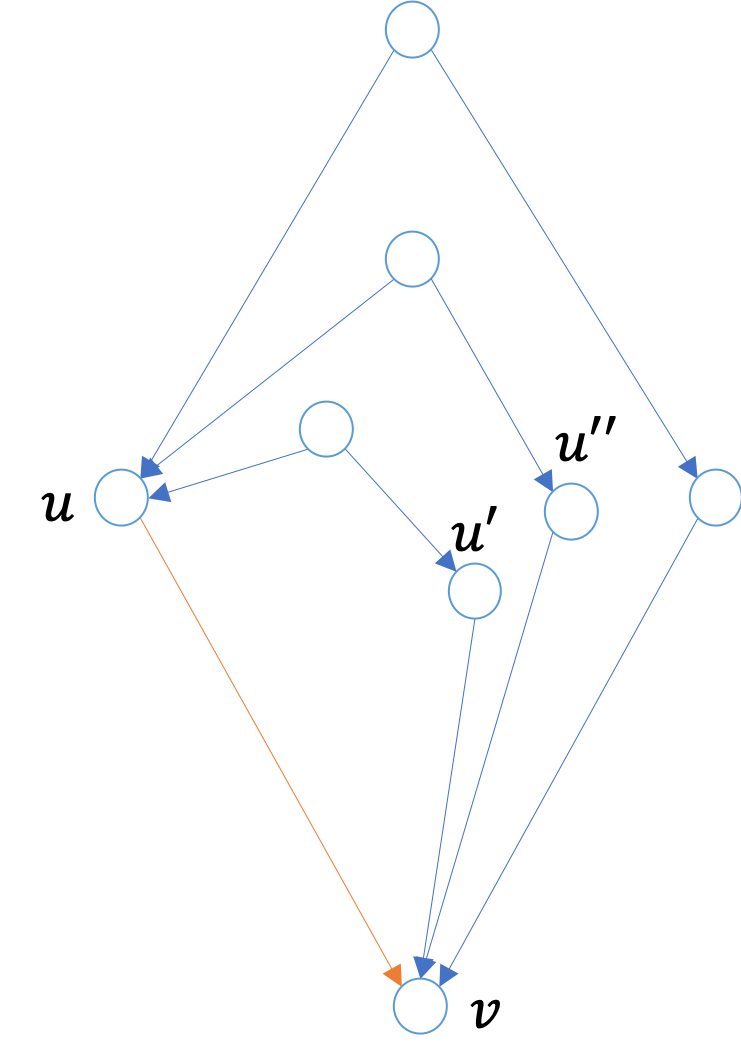}
\caption{The insertion of edge $(u, v)$ increases the rank of the boundary group by 3.}
\label{fig:negative}
\end{figure}

(Case-B): The endpoints $u$ and $v$ are already in the same connected component in $\oldG$. After adding this edge $e_s$, new cycles are created in $\newG$. Hence $e_s$ is \emph{positive} in this case (as it creates an element in $\mathsf{Z}_1$; although different from the standard simplicial homology, it may not necessarily create an element in $\HH_1$ as we will see later). 

Whether $e_s$ is positive or negative can be easily determined by performing two {\sf Find} operations in the union-find data structure representing ${T}_{s-1}$. A {\sf Union}($u,v$) operation is performed to update ${T}_{s-1}$ to ${T}_s$ if $e_s$ is negative. 

We now describe how to handle (Case-B). 
After adding edge $e_s$, multiple cycles containing $e_s$ can be created in $\newG$. Nevertheless, by Proposition \ref{prop:Z1}, the dimension of $\mathsf{Z}_1$ increases only by $1$. 
On the other hand, the addition of $e_s$ may create new boundary cycles. Interestingly, it could increase the rank of $\mathsf{B}_1$ by more than $1$. See Figure \ref{fig:negative} for an example where $rank(\mathsf{B}_1)$ increases by $3$; and note that this number can be made arbitrarily large.

As mentioned earlier, in this case, we wish to compute a set of generating boundary cycles $\candidateC_s$ such that $B \cup \candidateC_s$ contains a basis for $\mathsf{B}_1(\newG)$. 

Similar to Algorithm \ref{alg:slowAlg}, using Theorem \ref{thm:B1}, we choose some bigons, boundary triangles and boundary quadrangles and add them to $\candidateC_s$. 
In particular, since $\candidateC_s$ only accounts for the newly created boundary cycles, we only need to consider bigons, boundary triangles and boundary quadrangles that contain $e_s$. 
We now describe the construction of $\candidateC_s$, which is initialized to be $\varnothing$. 

(i) \emph{Bigons}. At most one bigon can be created after adding $e_s$ (namely, the one that contains $e_s$). We add it to $\candidateC_s$ if this bigon exists. 

\begin{figure}[H]
\centering
 \begin{subfigure}[t]{0.3\textwidth}
        \centering
       \includegraphics[width=3.5cm]{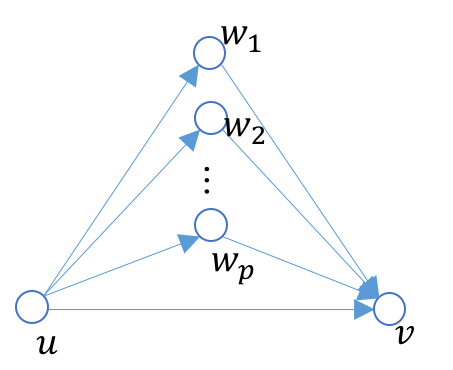}
        \subcaption{}
    \end{subfigure}
    ~ 
    \begin{subfigure}[t]{0.3\textwidth}
        \centering
       \includegraphics[width=3cm]{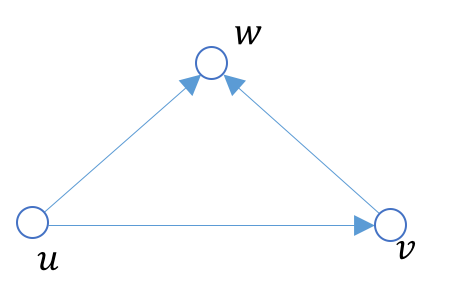}
        \subcaption{}
    \end{subfigure}
    ~ 
    \begin{subfigure}[t]{0.3\textwidth}
        \centering
       \includegraphics[width=3cm]{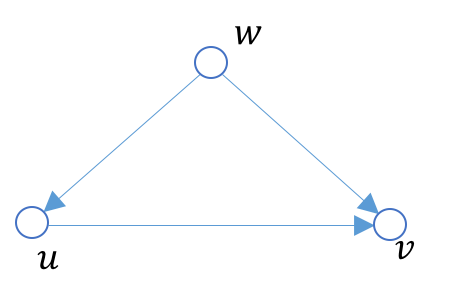}
        \subcaption{}
    \end{subfigure}
\caption{Three types of boundary triangles incident to $e_s = (u, v)$.}
\label{fig:persis_tri}
\end{figure}

(ii) \emph{Boundary triangles}. There could be three types of newly created boundary triangles containing $e_s = (u, v)$. 
The first case is when $u$ is the source and $v$ is the sink; see Figure~\ref{fig:persis_tri}(a). In this case multiple 2-paths may exist from $u$ to $v$, $e_{uw_1v}, e_{uw_2v},\cdots e_{uw_pv}$, forming multiple boundary triangles of this type containing $e_s$. However, we only need to add \emph{one} triangle of them into $\candidateC_s$, say $(u, v\mid w_1)$ since every other triangle $(u, v\mid w_j)$ can be written as a linear combination of $(u,v\mid  w_1)$ and an existing boundary quadrangle $(u, v\mid w_1, w_j)$ in $\oldG$. 

For the second case (see Figure~\ref{fig:persis_tri}(b)) where $u$ is the source but $v$ is not the sink, we include all such boundary triangles to $\candidateC_s$. 
We also add all boundary triangles of the last type in which $v$ is the sink but $u$ is not the source to $\candidateC_s$; see Figure~\ref{fig:persis_tri}(c). 
It is easy to see that $\candidateC_s \cup B$ can generate all new boundary triangles containing $e_s = (u, v)$. 

\begin{figure}[H]
\centering
\begin{subfigure}[t]{0.47\textwidth}
        \centering
       \includegraphics[width=5cm]{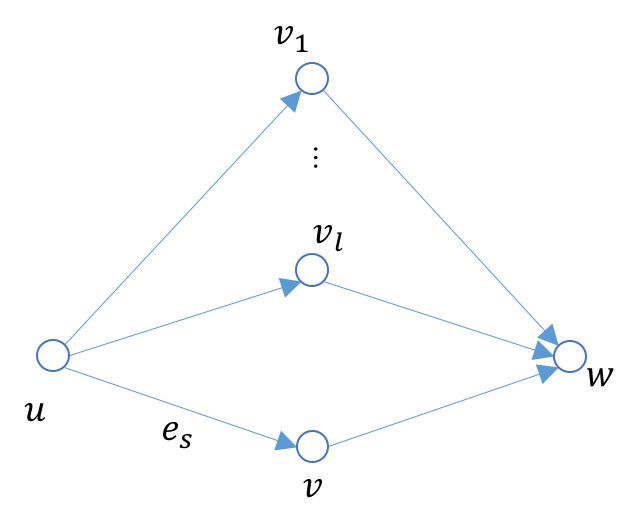}
        \subcaption{}
    \end{subfigure}
    ~ 
    \begin{subfigure}[t]{0.47\textwidth}
        \centering
       \includegraphics[width=5cm]{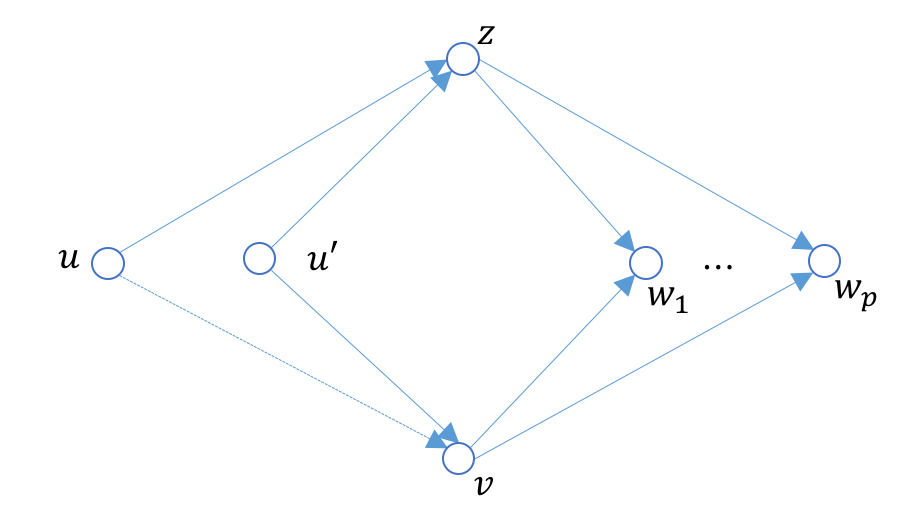}
        \subcaption{}
    \end{subfigure}

\caption{ (a) Examples of new  boundary quadrangles with $u$ being the source. (b) Not all boundary quadrangles in $M$ will be added to the generating set $\candidateC_s$. 
\label{fig:pers_quad}}
\end{figure}

(iii) \emph{Boundary quadrangles.} 
Given an edge $e_s=(u, v)$, there are two types of the boundary quadrangles incident to it: one has $u$ as the source; the other has $v$ as the sink. We focus on the first case; see Figure~\ref{fig:pers_quad}(a). The second case can be handled symmetrically.

In particular, we will first compute a set $M$ and then select a subset of quadrangles from $M$ for adding to $\candidateC_s$.

Specifically, take any successor $w$ of $v$, that is, there is an edge $(v, w) \in \oldG$ forming an allowed 2-path $e_{uvw}$ in $\newG$. Before introducing the edge $e_s$, there may be multiple allowed 2-paths $e_{uv_1w}, e_{uv_2w},\cdots, e_{uv_lw}$ in $\oldG$; see Figure \ref{fig:pers_quad} (a). 
For each such 2-path $e_{uv_kw}, 1\le k\le l$, a new boundary quadrangle $(u, w \mid v, v_k)$ containing $e_s = (u,v)$ will be created. 
However, among all such 2-paths $e_{uv_1w}, \cdots, e_{uv_lw}$, we will pick just one 2-path, say $e_{uv_1w}$
and only add the quadrangle $(u, w \mid v, v_1)$ formed by $e_{uvw}$ and $e_{uv_1w}$ to $M$. 
Observe that any other boundary quadrangle containing 2-path $e_{uvw}$, say $(u, w \mid v, v_k)$, can be written as a linear combination of the quadrangle $(u, w \mid v, v_1)$ and boundary quadrangle $(u, w \mid v_k, v_1)$ which is already in $\oldG$ (and in the span of $B$ which is a basis for $\mathsf{B}_1(\oldG)$). In other words, $(u, w \mid v, v_1) \cup B$ generates any other boundary quadrangle containing 2-path $e_{uvw}$. 

We perform this for each successor $w$ of $v$.  
Hence this step adds at most $d^{\oldG}_{out}(v)$ number of boundary quadrangles to the set $M$. 

Not all quadrangles in $M$ will be added to $\candidateC_s$. 
In particular, suppose we have $p$ quadrangles $A = \{(u, w_j \mid v, z):1\le j\le p\} \subseteq M$ incident to the newly inserted edge $e_s=(u, v)$ as well as another vertex $z$, i.e. there are edges $(u, z), (w_j, z)$ and $(v, w_j)$, $1\le j\le p$; see Figure~\ref{fig:pers_quad} (b). 
If there does not exist any other vertex $u'$ such that edges $(u', z), (u', v) \in \oldG$, then we add \emph{all} quadrangles in $A$ to $\candidateC_s$. If this is not the case, let $u'$ be another vertex such that $(u', z)$ and $(u', v)$ are already in $\oldG$; see Figure \ref{fig:pers_quad} (b). 
In this case, we only add \emph{one} quadrangle from set $A$, say, $(u, w_1|v, z)$ to the generating set $\candidateC_s$. 

It is easy to check that any other quadrangle $(u, w_j\mid v, z)$, $1 < j \le p$, can be written as the combination of $(u, w_1\mid v, z)$, $(u', w_1\mid v, z)$ and $(u', w_j\mid v,  z)$. As the latter two quadrangles are boundary quadrangles from $\oldG$, they can already be generated by $B$. 
The entire process takes time $O(|M|) = O(d^{\oldG}_{out}(v))$. 
It is also easy to see that $B \cup \candidateC_s$ can generate any boundary quadrangle containing $e_s=(u,v)$ and with $u$ being its source. 

The case when $v$ is the sink of a boundary quadrangle is handled symmetrically in time $O(d^{\oldG}_{in}(u))$. 
Hence the total time to compute a generating set $\candidateC_s$ is $O(d^{\oldG}_{in}(u) + d^{\newG}_{out}(v))$ when inserting a single edge $e_s=(u,v)$.

\subsubsection{Procedure \texorpdfstring{{\sc{FindPairs}}}{Lg}} 
Given the generating set $\candidateC_s$ and the previous basis $B$ for $\mathsf{B}_1(\oldG))$, we know that $\candidateC_s \cup B$ generate the new boundary group $\mathsf{B}_1(\newG)$. 
We now extract a basis $B_{new}$ for $\mathsf{B}_1(\newG)$ from $B \cup \candidateC_s$. 

We represent each 1-cycle $\gamma$ by an $m$-dimensional vector, also denoted by $\gamma$, so that $\gamma  = \sum_{i=1}^m \gamma[i] e_i$. (Note that $e_1, \ldots, e_m$ are sorted according to their filtration order.) 
A set of $k$ cycles can now be viewed as a $m\times k$ matrix, where the $i$-th column corresponds to the vector representation of the $i$-th cycle. 

Thus columns in the matrix $B$ correspond to cycles in an existing basis for $\mathsf{B}_1(\oldG)$), and are already linearly independent. 
Our goal is now to compute a basis of the form $[B \mid B']$ for the matrix $[B \mid \candidateC_s]$, and the columns in $B$ and $B'$ form a new basis $B_{new}$ for $\mathsf{B}_1(\newG)$. 

To do so, we follow the standard persistence algorithm which would also output persistence-pairings for $\HH_1$. 
Specifically, Let $low(j)$ be the row index of the last non-zero entry in column $j$ in a matrix $A$. A matrix $A$ is in \emph{reduced form} if the $low(j)$ in each column $j$ is unique. 
We compute persistent pairs by always maintaining the basis in reduced form \cite{cohen2006vines}. Here assume that $B$ is already in reduced form. 
We then perform standard column reduction to convert $[B |\candidateC_s]$ into reduced form $[B | R]$. 
For each non-zero column $R[j]$ in $R$, let $k=low(j)$ be the index of its lowest non-zero entry. Let $e_k$ be the edge corresponding to this entry in the cycle corresponding to column $R[j]$. This means that the cycle corresponding to $R[j]$ is created when $e_k$ introduced, but killed when introducing $e_s$, since currently it is a boundary. Then we add the persistence pairing $(\weightf(e_k), \weightf(e_s))$ to the output persistence diagram $Dg_1 G$ for the $1$-dimensional path homology. 

The collection of non-zero columns in $R$ gives rise to $B'$. Afterwards, we update $B$ to be $B \cup B'$, and proceed to process the next edge $e_{s+1}$.

\subsection{Analysis of Algorithm \ref{alg:fastAlg}}
\label{subsec:persistH1analysis}

\paragraph{Correctness.} 
Notice that the invariant that $B$ is a basis for $\newG$ at the end of the for-loop (line-7 of Algorithm \ref{alg:fastAlg}) is maintained. 
Furthermore, $B$ is always in reduced form which is maintained via left-to-right column additions only. Hence the algorithm computes the $1$-dimensional persistent path homology correctly \cite{cohen2006vines}.

\paragraph{Time complexity analysis.} 
The remainder of this section is devoted to determining the time complexity of Algorithm \ref{alg:fastAlg}. 
Specifically, we first show the following theorem. 
\begin{theorem}\label{thm:candidateSize}
Across all stages $s \in [1, m]$, the total cardinality of the generating set $\mathsf{C}=\cup_s \candidateC_s$ is $O(\csize)$. The total time taken by procedure {\sc NewBasis}($s$) for all $s\in [1, m]$ is $O(m + \sum_{(u, v)\in E}(d_{in}(u)+d_{out}(v))\})$. 
\end{theorem}

\begin{proof}
We will count separately the number of bigons, boundary triangles, and boundary quadrangles added to any $\candidateC_s$. 
Set $r = \csize$.

(i) \emph{Bigons:} First, it is easy to see that for each edge $e_s=(u, v)$ with $s\in [1, m]$, at most one bigon (incident to $e_s$) is added. Besides, if $\dout(v)=0$, there is no bigon incident to $e_s$.
Hence the total number ever added to $\mathsf{C}$ is $O(\min\{m, \sum_{(u, v)\in E}(\dout(v))\})=O(r)$ and it takes $O(m$) time to compute them.

(ii) \emph{Boundary triangles:} For boundary triangles, we know from Proposition \ref{prop:enumeration} that there are altogether $O(\arbor(G)m)$ triangles (thus at most $O(\arbor(G)m)$ boundary triangles) in a graph $G$ and they can all be enumerated in $O(\arbor(G)m)$ time. 
Obviously, the number of boundary triangles ever added to $\mathsf{C}$ is at most $O(\arbor(G)m)$.

We now argue that the number of boundary triangles added to $\mathsf{C}$ is also bounded by $O($ $\sum_{(u,v)\in E} (d_{in}(u) + d_{out}(v)))$. Note that, for every 2-path, at most one boundary triangle is added to the set. Since the number of 2-paths is indeed $\Theta(\sum_{(u,v)\in E} (d_{in}(u) + d_{out}(v)))$, the number of triangles we add is $O(\sum_{(u,v)\in E} (d_{in}(u) + d_{out}(v)))$.
Recall that there are three cases for boundary triangles added; see Figure~\ref{fig:persis_tri}. The time spent for the first case for every $s$ is $O(1)$by recording any 2-path $e_{uwv}$, and $O(d_{in}(u) + d_{out}(v))$ for the last two cases. Thus the total time spent at adding boundary triangles incident to $e_s$ and identifying triangles to be added to $\candidateC_s$ for all $s\in [1, m]$ takes $O(m+\sum_{(u,v)\in E} (d_{in}(u) + d_{out}(v)))$ time. 

(iii) \emph{Boundary quadrangles:} 
The situation here is somewhat opposite to that of the boundary triangles: Specifically, it is easy to see that this step accesses at most $O(d_{in}(u) + d_{out}(v))$ boundary quadrangles when handling edge $e_s = (u, v)$. Hence the number of boundary quadrangles it can add to $\candidateC_s$ is at most $O(d_{in}(u) + d_{out}(v))$. The total number of boundary quadrangles ever added to $\mathsf{C}$ is thus bounded by $O(\sum_{(u,v)\in E} (d_{in}(u) + d_{out}(v)))$. 

We now prove that the number of boundary quadrangles ever added to $\mathsf{C}$ is also bounded by $O(\arbor(G)m)$. 
We use the existence of a succinct representation of all quadrangles as specified in Proposition \ref{prop:enumeration} to help us argue this upper bound. Notice that our algorithm {\bf does not} compute this representation. It is only used to provide this complexity analysis. 

Specifically, by Proposition \ref{prop:enumeration}, we can compute a list $L$ of \mytriple{}s with $O(\arbor(G)m)$ total size complexity, which generates all undirected quadrangles. 
Following the proof of Theorem \ref{thm:candidate1}, we can further refine this list, where each \mytriple{} $\atriple\in L$ further gives rise to three lists that are of type-1, 2, or 3. 
Let $\widehat L$ denote this refinement of $L$, consisting of lists of type-1, 2 or 3. 
From the proof of Theorem \ref{thm:candidate1}, we know that the total size complexity for all lists in $\widehat L$ is still $O(\arbor(G)m)$. This also implies that the cardinality of $\widehat{L}$ is bounded by $|\widehat{L}| = O(\arbor(G)m)$. 

We now denote by $\mathcal{R}$ the set of all boundary quadrangles ever added to $\candidateC = \cup_s \candidateC_s$ by Algorithm \ref{alg:fastAlg}. 
Furthermore, let 
\[\mathsf{P} := \{ (\atriple, w) \mid \atriple \in \widehat{L}, w \in \atriple \}. \]
Below we show that we can find an injective map $\pi: \mathcal{R} \to \mathsf{P}$. 
But first, note that $|\mathsf{P}|$ is proportional to the total size complexity of $\widehat L$ and thus is bounded by $O(\arbor(G) m)$. 

We now establish the injective map $\pi: \mathcal{R} \to \mathsf{P}$. 
Specifically, we process each boundary quadrangle \emph{in the order} that they are added to $\candidateC$. 
Consider a boundary quadrangle $R = R(u, v, w, z)$ added to $\candidateC_s$ while processing edge $e_s = (u, v)$. 
There are two cases: The first is that $R$ is of the form $(u, w \mid v, z)$ in which $u$ is the source of this quadrangle. The second is that it has the form $(w, v \mid u, z)$ in which $v$ is the sink. 
We describe the map $\pi(R)$ for the first case, and the second one can be analyzed symmetrically. 

By construction of $L$, there is at least one \mytriple{} $\atriple \in L$ covering $R= (u, w \mid v, z)$. There are three possibilities: 

(Case-a): The \mytriple{} $\atriple$ is of the form $\atriple = \atriple_{uw} = (u, w, \{ \cdots \}).$ In this case, the boundary quadrangle $R = (u, w \mid v, z)$ is in a type-1 list $\atriple^{(1)}_{uw} = (u, w, S) \in \widehat{L}$, and both $v, z \in S$.
We now claim that the pair $(\atriple^{(1)}_{uw}, v) \in \mathsf{P}$ has not yet been mapped (i.e, there is no $R' \in C$ with $\pi(R') = (\atriple^{(1)}_{uw}, v)$ yet), and we can thus set $\pi(R) = (\atriple^{(1)}_{uw}, v) \in \mathsf{P}$. 
Suppose on the contrary there already exists $R' \in C$ that we processed earlier than $R$ with $\pi(R') = (\atriple^{(1)}_{uw}, v)$. In that case, $R' = (u, w \mid v, z')$ must contain the 2-path $e_{uvw}$ as well. Since $R'$ is processed earlier than $R$, and edge $e_s = (u, v)$ is the most recent edge added, $R'$ must be added when we process $e_s$ as well (as $R'$ contains $e_s$). 
However, Algorithm \ref{alg:fastAlg} in this case only adds \emph{one} quadrangle containing the 2-path $e_{uvw}$, meaning that $R'$ cannot exist (as otherwise, we would not have added $R$ to $\candidateC_s$; recall Figure \ref{fig:pers_quad} (a)). 
Hence, the map $\pi$ so far remains injective. 

(Case-b): The \mytriple{} $\atriple$ is of the form $\atriple = \atriple_{wu} = (w, u, \{ \cdots \}).$
In this case, this quadrangle is covered by the type-2 list $\atriple^{(2)}_{wu} \in \widehat{L}$. We handle this in a manner symmetric to (Case-a) and map $\pi(R) = (\atriple^{(2)}_{wu}, v)$. 

(Case-c): The last case is that $R$ is generated by \mytriple{} $\atriple$ of the form $\atriple_{vz} = (v, z, \{\cdots \})$. In this case, the quadrangle $R = (u,w \mid v, z)$ will be covered by the type-3 list $\atriple^{(3)}_{vz} = (v, z, S_1, S_2)$ with $u \in S_1$ and $w\in S_2$; see Figure \ref{fig:prof_map}. 
We now argue that at least one of $(\atriple^{(3)}_{vz}, u)$ and $(\atriple^{(3)}_{vz}, w)$ has \emph{not} been mapped under $\pi$ yet. 
\begin{figure}[H]
\includegraphics[width=8cm]{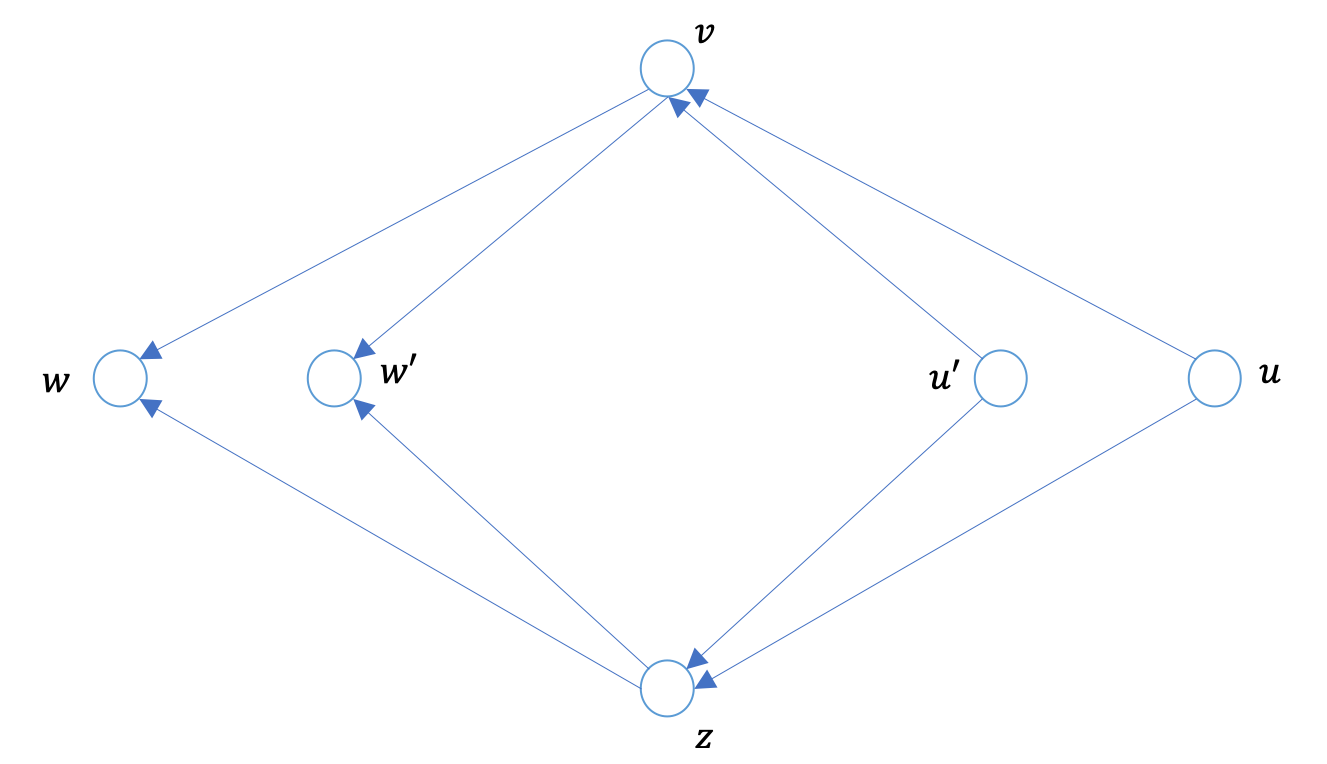}
\centering
\caption{At least one of $(\xi^{(3)}_{vz}, u)$ and $(\xi^{(3)}_{vz}, w)$ has not been mapped yet.}
\label{fig:prof_map}
\end{figure}

Suppose this is not the case and we already have both $\pi(Q_1) = (\atriple^{(3)}_{vz}, u)$ and $\pi(Q_2) = (\atriple^{(3)}_{vz}, w)$. Then $Q_1$ is necessarily of the form $(u, w' \mid v, z)$ and $Q_2$ is of the form $(u', w \mid v, z)$; and both $Q_1$ and $Q_2$ are processed before $R$. See Figure \ref{fig:prof_map}. 
Furthermore, $Q_1$ is only added when we process edge $e_s$. 
However, in this case, once $Q_1$ is added, Algorithm \ref{alg:fastAlg} will not add further quadrangle containing edges $(u, v)$ and $(u,z)$ (recall the handling of Figure \ref{fig:pers_quad} (b)). 
Hence $R$ cannot be added to $\candidateC_s$ in this case. 

In other words, it cannot be that both $Q_1$ and $Q_2$ already exist, and hence we can set $\pi(R)$ to be one of $(\atriple^{(3)}_{vz}, u)$ and $(\atriple^{(3)}_{vz}, w)$ that is not yet mapped. 
Consequently, the map $\pi$ we construct remains injective. 

We process all quadrangles in $\candidateC$ in order. The final $\pi: \mathcal{R} \to \mathsf{P}$ is injective, meaning that $|\mathcal{R}| \le |\mathsf{P}|$ and thus $|\mathcal{R}| = O(\arbor(G) m)$. 

Putting everything together, we have that the total number of boundary quadrangles added to $\candidateC$ is bounded by $O(\min \{ \arbor(G)m,  \sum_{(u,v)\in E} (d_{in}(u) + d_{out}(v)) \})$. 

Finally, Algorithm \ref{alg:fastAlg} spends $O(m+\sum_{(u,v)\in E} (d_{in}(u) + d_{out}(v)))$ time to handle both cases in Figure \ref{fig:pers_quad}.
The theorem then follows. 
\end{proof}

We now prove our main result, Theorem \ref{thm:main}. 
In particular, below we first show that Algorithm \ref{alg:fastAlg} takes $O(r m^2)$ time to compute the $1$-dimensional \emph{persistent} path homology, where 
\[r = \csize.\]
We then describe how to improve the time complexity to $O(r m^{\omega-1})$ to establish Theorem \ref{thm:main}.

Specifically, by using a union-find data structure to maintain a spanning forest for $G^{(s)}$, it takes $O(m\alpha (n))$ time, where $\alpha(\cdot)$ is the inverse Ackermann function, to check whether each edge $e_s$ is negative or positive (forming a cycle or not). 
As explained in the algorithm, if it is negative, then nothing needs to be done. If it is positive, then we need to compute a generating set as well as find new basis for the boundary group and compute persistence. 

By Theorem \ref{thm:candidateSize}, the total time complexity spent on all executions of line 5 of Algorithm \ref{alg:fastAlg} is $O( m+ \sum_{(u,v)\in E} (\din(u) + \dout(v)) )$. We now argue that this value is bounded by $O(r m)$. This clearly holds if $r = \sum_{(u,v)\in E} (\din(u) + \dout(v))$. If on the other hand $r = \arbor(G) m$, then $r = \Omega(m)$ and since $\sum_{(u,v)\in E} (\din(u) + \dout(v)) = O(mn)$, we have that 

\[ m + \sum_{(u,v)\in E} (\din(u) + \dout(v)) = O(r) + O(mn)  = O(rm). \] 

Next we bound the total time taken by all executions of line 6 of Algorithm \ref{alg:fastAlg} (i.e, on calling procedure {\sc FindPairs}($s$) for all $s \in [1,m]$). 
Note that at the $s$-th stage, we need to reduce $|\candidateC_s|$ number of columns in a matrix of the form $[B \mid \candidateC_s]$, each of which is of length $m$. Furthermore, since when reducing a specific column in a left-to-right manner, the number of non-zero columns to its left is bounded by the size of $rank(\mathsf{B_1}(G^{(s)})) = O(m)$, each column-reduction takes $O(m^2)$ time. 
As $|\candidateC| = \sum_{s} |\candidateC_s| = O(r)$ by Theorem \ref{thm:candidateSize}, procedure {\sc FindPairs}($s$) for all $s \in [1,m]$ takes $O(rm^2)$ time. 
It follows that Algorithm \ref{alg:fastAlg} takes time $O(rm + r m^2) = O(rm^2)$.

\vspace*{0.1in} \noindent \emph{Improving the time complexity.} 
We see that the dominating term for the time complexity of Algorithm \ref{alg:fastAlg} is the time spent on procedure {\sc FindPairs}($s$) for all $s \in [1,m]$. 
To this end, instead of reducing the columns for each $s = 1, \ldots, m$, we compute the earliest basis from the matrix $[ \candidateC_1 \mid \candidateC_s \mid \cdots \mid \candidateC_m]$ by the algorithm in~\cite{LSPEarly}. Here the earliest basis of a matrix $A$ with rank $r$ means the set of columns $B_{opt} = \{a_{i_1}, \cdots, a_{i_r}\}$ if the column indices $\{i_1,\cdots , i_r\}$ are the lexicographically smallest index set such
that the corresponding columns of $A$ have full rank. 
Since $|\cup_s \candidateC_s| = O(r)$ and the number of  independent columns is $O(m)$, this can be done in $O( (\frac{r}{m}) m^\omega) = O(r m^{\omega-1})$ time. 
This finishes the proof of Theorem \ref{thm:main}.

\vspace*{0.1in}\noindent{\underline{\emph{Remark: 1}}}
We note that neither term in $r = \csize$ always dominates. In particular, it is easy to find examples where one term is significantly smaller (asymptotically) than the other. 
For example, for any planar graph $G$, $\arbor(G) m = O(n)$. However, it is easy to have a planar graph where the second term $\sum_{(u,v)\in E} (\din(u) + \dout(v) ) = \Omega(n^2)$; see e.g, Figure \ref{fig:quadratic}. 

On the other hand, it is also easy to have a graph $G$ where $\sum_{(u,v)\in E} (\din(u) + \dout(v) ) = O(1)$ yet $\arbor(G) m = \Theta(n^3)$. Indeed, consider the bipartite graph in Figure \ref{fig:nobndQ}, where for each edge $(u,v)\in E$, $\din(u) + \dout(v) = 0$. 
However, this graph has $\arbor(G) = \Theta(n)$, $m = \Theta(n^2)$ and thus $\arbor(G)m = \Theta(n^3)$. 

\vspace*{0.1in}\noindent{\underline{\emph{Remark: 2: }}}
We note that the time complexity of the algorithm proposed by Chowdhury and M\'{e}moli in \cite{chowdhury2018persistent} to compute the $(d-1)$-dimensional persistence path homology takes $O(n^{3 + 3d})$ time. 
However, for the case $d=2$, a more refined analysis shows that in fact, their algorithm takes only $O((\sum_{(u, v)\in E} (d_{in}(u)+d_{out}(v))) mn^2)$ time. 

Compared with our algorithm, which takes time $O( r m^{\omega-1})$ with $r = \min\{ \arbor(G)m, \sum_{(u, v)\in E} ($ $d_{in}(u)+d_{out}(v))\}$ and $\omega < 2.373$, observe that our algorithm can be significantly faster (when $\arbor(G)m$ is much smaller than $\sum_{(u, v)\in E} (d_{in}(u)+d_{out}(v))$. 
For example, for planar graphs, our algorithm takes $O(n^{\omega})$ time, whereas the algorithm of \cite{chowdhury2018persistent} takes $O(n^5)$ time. 

\section{Applications}\label{sec:application}
 We first show  in Section \ref{subsec:minbasis} that our new algorithm can be extended to compute a minimal $1$-dimensional path homology basis. We then show in Section \ref{subsec:exp} some preliminary experimental results for our algorithms, including showing the efficiency of our algorithm compared to the previous best algorithm over several datasets. 

\subsection{Annotations and minimum \texorpdfstring{$1$}{Lg}-dimensional (path) homology basis}
\label{subsec:minbasis}

Our algorithm can also compute a ($1$-dimensional) minimal homology basis for a directed graph $G= (V,E)$.
In particular, let $g = rank(\HH_1)$, and assume that $G = (V, E)$ is equipped with positive edge weights $w: E \to \mathbb{R}^+$. 
Given any 1-cycle $\gamma = \sum c_i e_i$, for each $i\in [1, m]$ set $\hat{c}_i = 1$ if $c_i \neq 0$; and $\hat{c}_i = 0$ otherwise. Then the \emph{length of $\gamma$}, $\mu(\gamma)$, equals $\mu(\gamma) = \sum_i \hat{c}_i w(e_i)$. 

Now abusing the notations slightly, we say that a set of $g$ 1-cycles $\{ \gamma_1, \cdots, \gamma_g\}$  forms a $1$-dimensional homology basis if the homology classes they represent, $\{ [\gamma_1], \cdots, [\gamma_g]\}$, forms a basis for $\HH_1(G)$. 
The \emph{total length} of this homology basis equals the total lengths of all cycles involved, i.e. $\sum_{i=1}^g \mu(\gamma_i)$.

\begin{definition}
Given a weighted directed graph $G=(V, E)$, let $g$ be the rank of $1$-dimensional homology group $\HH_1(G)$. 
Let $\mu: \mathsf{Z}_1 \to \mathbb{R}^+\cup \{0\}$ be the length of each cycle $C \in \mathsf{Z}_1$. 
A homology basis $\gamma_1, \cdots, \gamma_g$ is called minimal if $\sum_{i=1}^g \mu(\gamma_i)$ is minimal among all bases of $\mathsf{H}_1$.
\end{definition}

\paragraph{Annotation}To this end, we first compute the so-called annotations of cycles to represent its homology class~\cite{Busaryev2012Annotating}. Then we use annotation to compute a minimal homology basis. 
In particular, for every 1-chain $\gamma\in \Omega_1$, we assign $\gamma$ a $g$-bit vector $\alpha(\gamma)$, called the \emph{annotation of $\gamma$}, such that any two cycles $C$ and $C'$ are homologous if and only if their annotations $\alpha(C)$, $\alpha(C')$ are the same. As before, set $r = \csize$. 

We compute the annotation according to \cite{Busaryev2012Annotating}. First we compute the annotation for every edge. We construct a cycle basis $Z$ in which any cycle can be expressed in simple and efficient terms. Note that as in the previous section, Algorithm~\ref{alg:fastAlg} can not only compute a basis for $\HH_1(G)$, it can in fact partition edges into negative edge set and positive edge set, in which all negative edges $E(T)$ form a spanning tree $T$, while every positive edge $e\in \hat{E}=E\setminus E(T)$ can create a new cycle together with $T$, denoted as $\gamma(T, e)$. Note that $\gamma(T, e)$ can be written as a $m$-bit vector. According to Proposition~\ref{prop:Z1}, all such cycle form a cycle basis $Z$. For every tree edge $e\in E(T)$, let $\gamma(T, e)$ be a $m$-bit vector where every element is 0.
It has been proved in~\cite{Busaryev2012Annotating} that for any cycle $C=\sum_{i}c_{i}e_{i}$ where $c_i$ is the coefficient of $e_i$ in $C$, it holds that $C=\sum_{i}c_{i}\gamma(T, e_i)$.
Note that Algorithm~\ref{alg:fastAlg} indeed computes a boundary basis $B$ for $\mathsf{B}_1(G)$ and a homology basis $H$ for $\HH_1(G)$, such that $B \cup H$ forms another cycle basis $\widehat Z$ for the cycle-group $\mathsf{Z}_1(G)$
where $rank(\hat Z)=rank(Z)=m-n+1$. Thus for every cycle $\gamma(T, e)$, 
we can solve a linear system $\hat Zx=\gamma(T, e)$. As a result, $x$ will be a $(m-n+1)\times 1$ vector. We assign last $g = rank(\HH_1)$-bits as the annotation $\alpha(e)$ of $e$. 
We can solve for the annotation for all edges in $O(m^{\omega})$ time where $m$ is the number of edges. Combined with time needed to compute the homology basis $H$, we conclude that computing the annotations of all edges costs $O(min\{\arbor(G)m, \sum_{(u, v)\in E} (d_{in}(u)+d_{out}(v))\} m^{\omega-1}+m^{\omega})$ time.

After computing the annotation for every edge, i.e. $\alpha(e)$ for every edge $e$, for every 1-cycle $C=\sum_{i}c_ie_i$ where $c_i$ is the coefficients for the edge $e_i$, its annotation can be computed as $\alpha(C)=\sum_{i}c_i\alpha(e_i)$. 

\paragraph{Minimal homology basis}We now compute a minimal homology basis. Using the results in~\cite{liebchen2005greedy}, we observe that there is a collection of cycles, called Horton family, that includes a minimal homology basis.  It is known from~\cite{horton1987polynomial} that the cardinality of the Horton family is $O(nm)$. 
Using same steps as in~\cite{liebchen2005greedy}, replacing vectors of cycles with their annotations, we can compute a $1$-dimensional minimal path homology basis.

According to \cite{Busaryev2012Annotating}, we have the following: (1) given annotations for edges, the time to compute the annotations of Horton cycles, is $O(n^2g+mn\log n)$; (2) given $O(nm)$ Horton cycles, extracting a minimal homology basis needs $O(nmg^{\omega-1})$ time. Thus in total the time to compute minimal homology basis is $O(r m^{\omega-1}+m^{\omega}+nmg^{\omega-1})$, where $r = \min\{\arbor(G)m, m+\sum_{(u, v)\in E} (d_{in}(u)+d_{out}(v))\}$.

\subsection{Preliminary experimental results}
\label{subsec:exp}
We implemented our algorithms to compute both the $1$-dimensional path homology and the $1$-dimensional persistent path homology. 
The implementaion is in Python 3 using the coefficient field $\mathbb{F}=\mathbb{R}$. 
\subsubsection{Comparing our algorithm with previous algorithm}
We test on some datasets on both our software
\footnote{Code is available in
\hyperlink{https://github.com/tianqicoding/1dPPH}{https://github.com/tianqicoding/1dPPH}} as well as the software in~\cite{chowdhury2018persistent}
\footnote{Code can be found in \hyperlink{https://github.com/samirchowdhury/pypph}{https://github.com/samirchowdhury/pph-matlab}}, including U.S. economic sector data, cycle network, C.elegans, Citeseer and Cora.
All experiments are worked on a PC Macbook Pro with Processor 2.9GHz Intel Core i5, Memory 8GB 1867 MHz DDR3.
We first introduce the datasets, and then analyse the output.

\paragraph{U.S. economic sector data} This dataset is released by the U.S. Bureau of Economic Analysis(\hyperlink{https://www.bea.gov/}{https://www.bea.gov/}) of the ``make'' and ``use'' tables of the production of commodities by industries. We obtained ``use'' table data for 15 industries across the year range 1997-2015; we have 19 tables, each showing the yearly asymmetric flow of commodities across industries. We used same preprocessing used by~\cite{chowdhury2018persistent}, after which we get 19 complete directed graphs.

\paragraph{Asian migration and remittance} This dataset is the Asian net migration and remittance networks in 2015  including 50 countries and regions, obtained from UN Global Migration Data\-base and on bilateral remittances from the World Bank database as reported respectively in~\cite{migration} and \cite{remittance}. \cite{ignacio2019tracing} analyzed the data via directed clique complex. Here we use the same preprocessing as~\cite{ignacio2019tracing} and work on persistent path homology.

\paragraph{Directed cycle graph with 1000 vertices} This dataset is a directed cycle graph with 1000 vertices. Each edge follows in the same direction; for every vertex, both the indegree and outdegree are exactly 1.

\paragraph{C.elegans} This dataset is a chemical synapse network of C.elegans~\cite{varshney2011structural}. It is a directed graph with 279 vertices and 2194 edges. Each vertex represents a neuron, and every edge reflects the synaptic contact between corresponding neurons; the source is the sender, and the sink is the receiver.
Every edge has an weight, meaning the number of chemical synapse the receiver received from the sender.

\paragraph{Citation networks--Citeseer and Cora~\cite{sen2008collective}} These two datasets are both citation benchmarks. Both are directed networks, indicating the citation between papers. There are 2708 vertices and 5429 edges in Cora, while 3279 vertices and 4608 edges in Citeseer.

\begin{table}[h]
  \begin{tabular}{ |l | c | r |}
    \hline
    Dataset Name  & Our algorithm(s) & Algorithm in~\cite{chowdhury2018persistent}(s)\\ \hline
    U.S.economic sector data(19 graphs) &0.7381(average) &0.4795(average)\\ \hline
    Migration&2.1277&3.8396\\
    \hline
    Remittance&1.7780&2.0914\\\hline
    Directed cycle graph&0.0047&0.8197\\ \hline
    C.elegans & 21.2429 &71.1714\\ \hline
    Cora  &6.0170&27.3146 \\ \hline
    Citeseer&2.6352&13.5158\\\hline
    \end{tabular}
    \caption{}
   \label{tab:comp}
   \vspace{-20pt}
  \end{table}
\paragraph{Analysis}Table~\ref{tab:comp} summerizes the runtimes for our approach as well as the algorithm in~\cite{chowdhury2018persistent}.
It implies that the baseline approach has a better performance than ours for small but dense graphs. This is because in dense graphs(e.g. complete graphs), the candidate set we compute in our algorithm is comparable with the size of 2-paths; the term $r\approx\sum_{(u, v)\in E}(d_{in}(u)+d_{out}(v))$ in Theorem~\ref{thm:main}. 

However, our algorithm significantly outperforms the baseline approach for graphs whose edges and vertices are comparable. The reason is that in this case, the term $r$ is determined by $\arbor (G)m$. In this type of graphs, $\arbor (G)m$ is smaller than the number of 2-paths.
Furthermore, our algorithm performs on large dataset well which makes (persistent) path homology more applicable in real datasets.

\subsection{Asian Migration and Remittance}
The datasets show the migration and remittance between countries and regions in asia. We list all countries and regions in Table~\ref{tab:contries}.
\begin{table}[h]
  \begin{tabular}{llllll}
  \hline
Order &  Country &  Abbrev.&Order &  Country &  Abbrev.\\\hline
1& Afghanistan & AF &26&Lebanon& LB\\\hline
2 & Armenia& AM&27&Malaysia& MY\\\hline
3 & Azerbaijan& AZ&28& Maldives& MV\\\hline
4 & Bahrain& BA&29&Mongolia &MN\\\hline
5 & Bangladesh& BD&30&Myanmar &MM\\\hline
6&  Bhutan& BT&31&Nepal &NP\\\hline
7&  Brunei Darussalam& BN &32&Oman &OM \\\hline
8&  Cambodia& KH&33&Pakistan &PK\\\hline
9&  China(mainland)& CN&34&Philippines &PH\\\hline
10&  Hong Kong& HK&35&Qatar &QA\\\hline
11&  Macau& MO&36&Republic of Korea &KR\\\hline
12&  Cyprus& CY&37&Saudi Arabia &SA\\\hline
13&  Dem. People’s Rep. of Korea& KP&38&Singapore &SG\\\hline
14&  Georgia& GE&39&Sri Lanka &LK\\\hline
15&  India& IN&40&State of Palestine &PS\\\hline
16&  Indonesia& ID&41&Syria &SY\\\hline
17&  Iran& IR&42&Tajikistan &TJ\\\hline
18&  Iraq& IQ&43&Thailand& TH\\\hline
19&  Israel& IL&44&Timor-Leste &TI\\\hline
20&  Japan& JP&45&Turkey &TR \\\hline
21&  Jordan& JO&46&Turkmenistan &TM\\\hline
22&  Kazakhstan&  KZ&47&United Arab Emirates& AE\\\hline
23&  Kuwait& KW&48&Uzbekistan &UZ\\\hline
24&  Kyrgyzstan& KG&49&Vietnam &VN\\\hline
25&  Laos& LA&50&Yemen &YE\\\hline
 \end{tabular}
    \caption{}
   \label{tab:contries}
   \vspace{-20pt}
  \end{table}
                      
The persistent path homologies of migration and remittance networks are motivated by flow structure. We are concerned with flows containing single sink and single source. Different with~\cite{ignacio2019tracing}, we regard cycles not only with boundary triangle structure but also with boundary quadrangle structure as trivial: two flows with one or two edges from the source to the sink are equivalent. Under this motivation, we capture those non-trivial cycles, as well as compute the persistence.

We first do some preprocess on the dataset as in~\cite{ignacio2019tracing}. For remittance data, let $r_{ab}$ be the remittance from $a$ to $b$. We create an edge $(a, b)$ if $r((a, b))=r_{ab}-r_{ba}>0$; there are no bi-gons in the graph. We set the edge weight $w(e)$ to  $w(e)=max_{e}(r(e))+1-r(e)$, inducing a filtration of path homologies.  The edge weights transform largest remittance to smallest. As a consquence, cycles with large weights are born early, and cycles that are born early but killed off by edges with smaller weights will have large persistence. Same preprocessing is employed on migration dataset.

We process persistent path homology on migration network. 
There are 44 generating cycles while there are 61 in~\cite{ignacio2019tracing}, meaning that there are 17 generating homology class non-trivial in~\cite{ignacio2019tracing} become trivial now when they are born because of the boundary quadrangles: there are two 2-path flows from the source to the sink. 
Same as ~\cite{ignacio2019tracing}, all cycles are killed of at the end, indicating small migration flow across countries and regions within cycles. 
We listing all generating cycles as Figure~\ref{fig:migration1} and Figure~\ref{fig:migration2}. Here all red nodes denote sources, greens denote sinks and blues denote other vertices. Figure~\ref{fig:migration1} and Figure~\ref{fig:migration2} reflect that neighbor countries tend to attract immigrants from the same countries, e.g. there are big migrations from India to both Kuwait and United Arab Emirates. On the other hand, people in neighbor countries, e.g. there are amount of people in both China(mainland) and Philippines moving to Hong Kong and Japan. Besides, we capture a directed cycle Kazakhstan$\to$Kyrgyztan$\to$ Tajikistan $\to$Kazakhstan, indicating a directed flow.

The generating cycles for remittance network are listed in Figure~\ref{fig:migration1} and Figure~\ref{fig:migration2}. Most generating cycles are consistent with cycles in migration generating cycles: the generating cycle for the remittance network is a generating cycle in migration network with all arrows reversed. For example, there is a generating cycle in Figure~\ref{fig:migration1}: India$\to$United Arab Emirates$\gets$Indonesia$\to$Saudi Arabia$\gets$India. It corresponds to a generating cycle in Figure~\ref{fig:remittance1}: India$\gets$United Arab Emirates$\to$
Indonesia$\gets$Saudi Arabia$\to$India. The guess behind that could be after migration, people tend to send money back to home.

\begin{figure}[H]
    \centering
    \includegraphics[width=14cm]{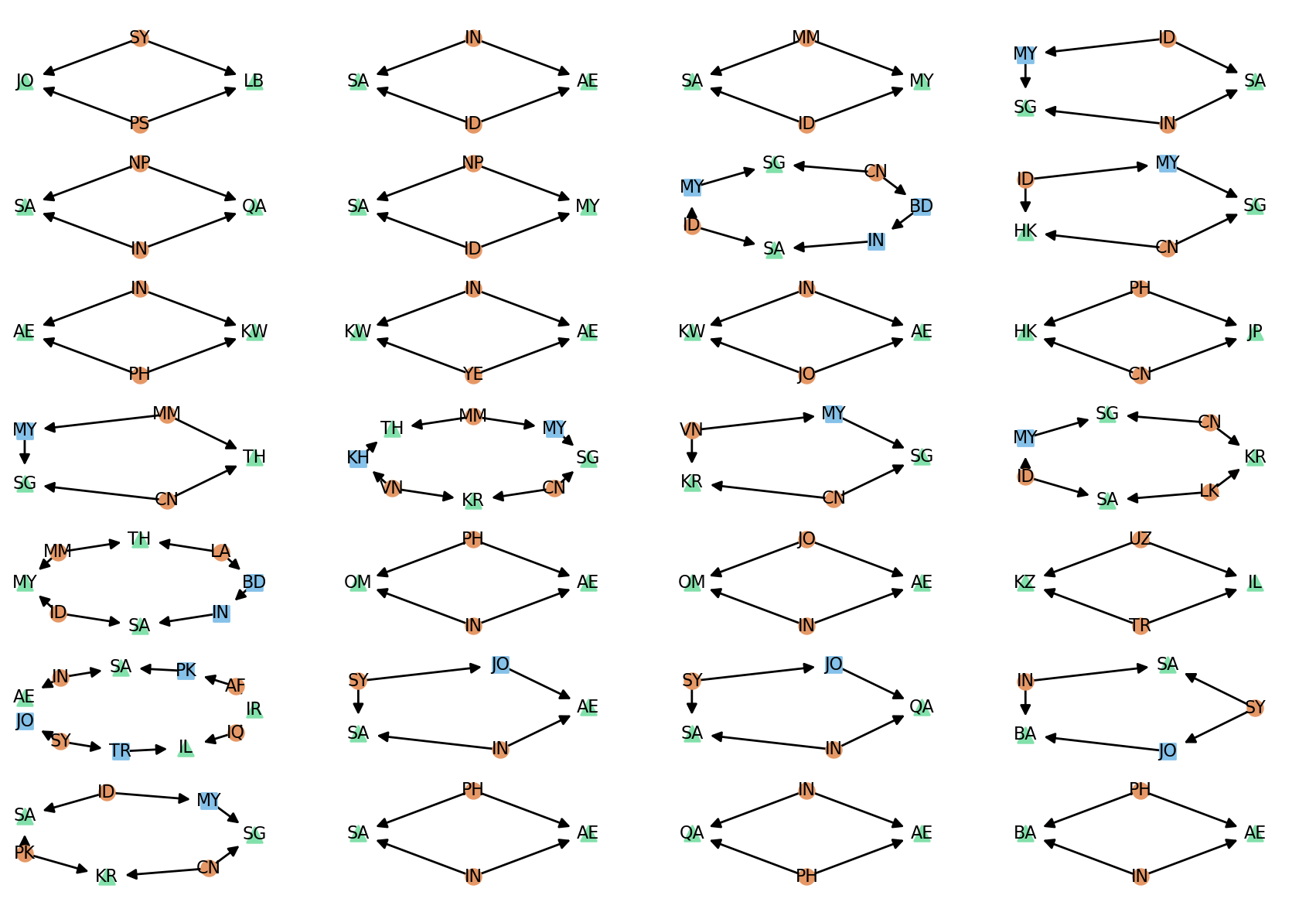}
    \caption{Generating cycles for persistent path homology on migration network}
    \label{fig:migration1}
    \centering
    \includegraphics[width=14cm]{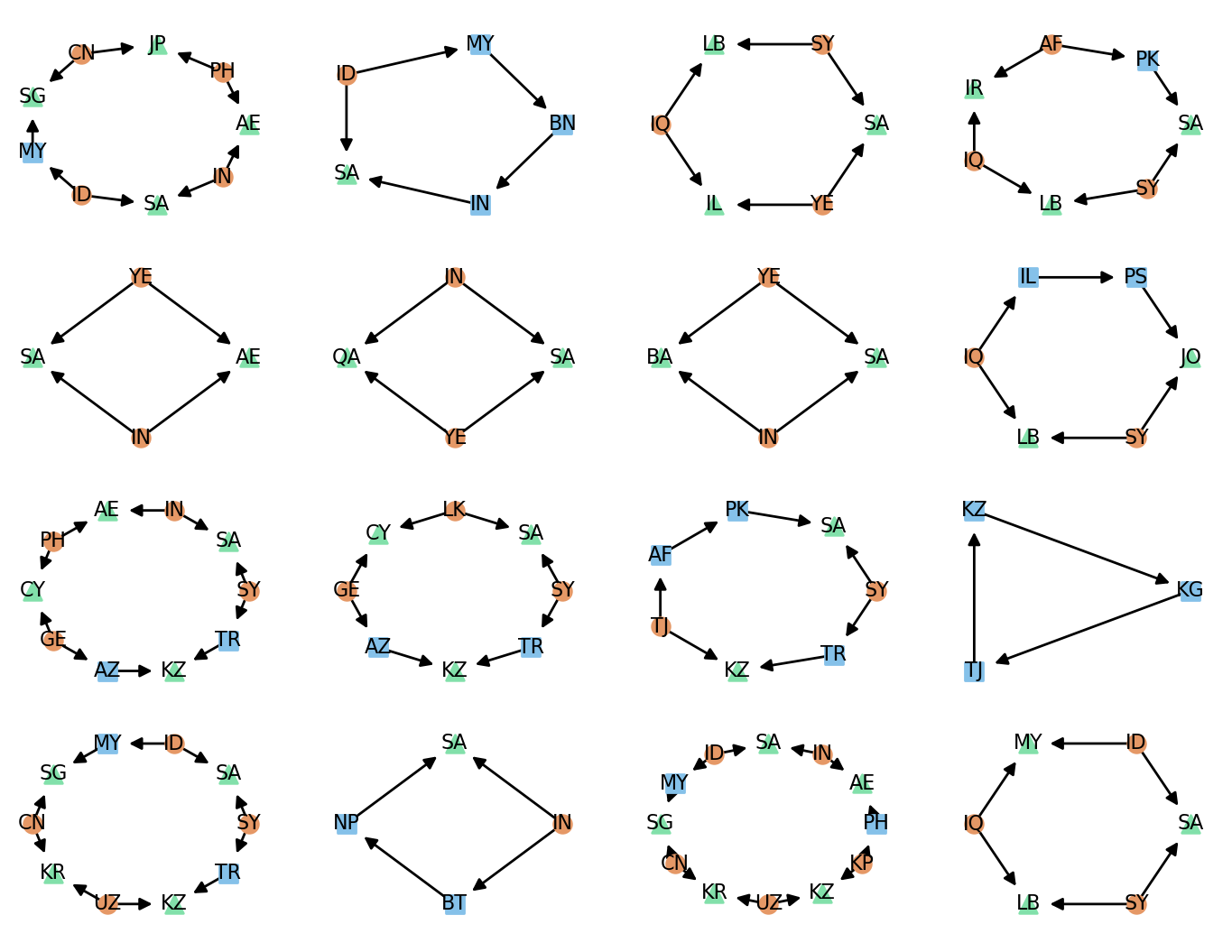}
    \caption{Generating cycles for persistent path homology on migration network(cont.)}
    \label{fig:migration2}
\end{figure}
\begin{figure}[H]
    \centering
    \includegraphics[width=14cm]{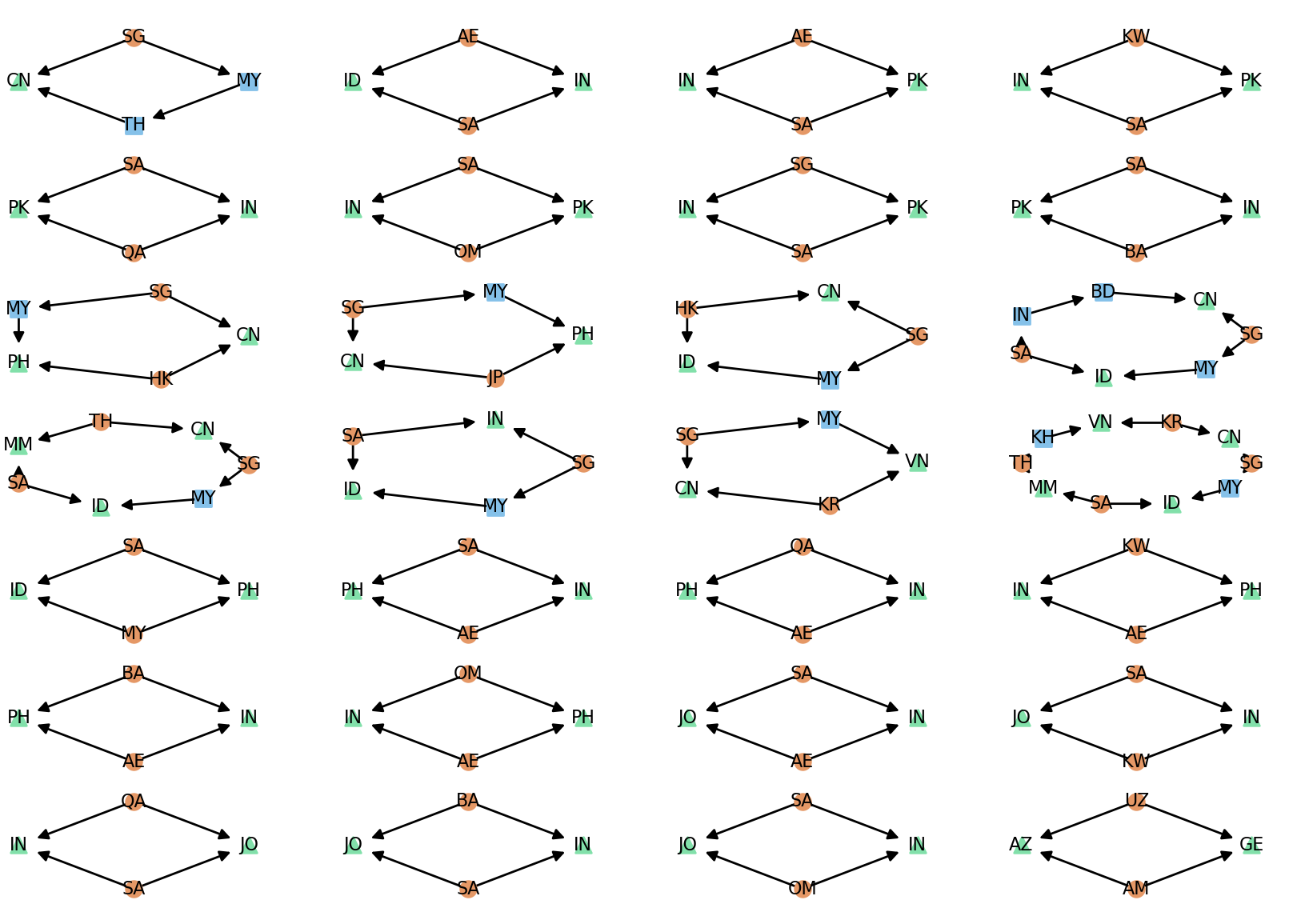}
    \caption{Generating cycles for persistent path homology on remittance network}
    \label{fig:remittance1}
    \centering
    \includegraphics[width=14cm]{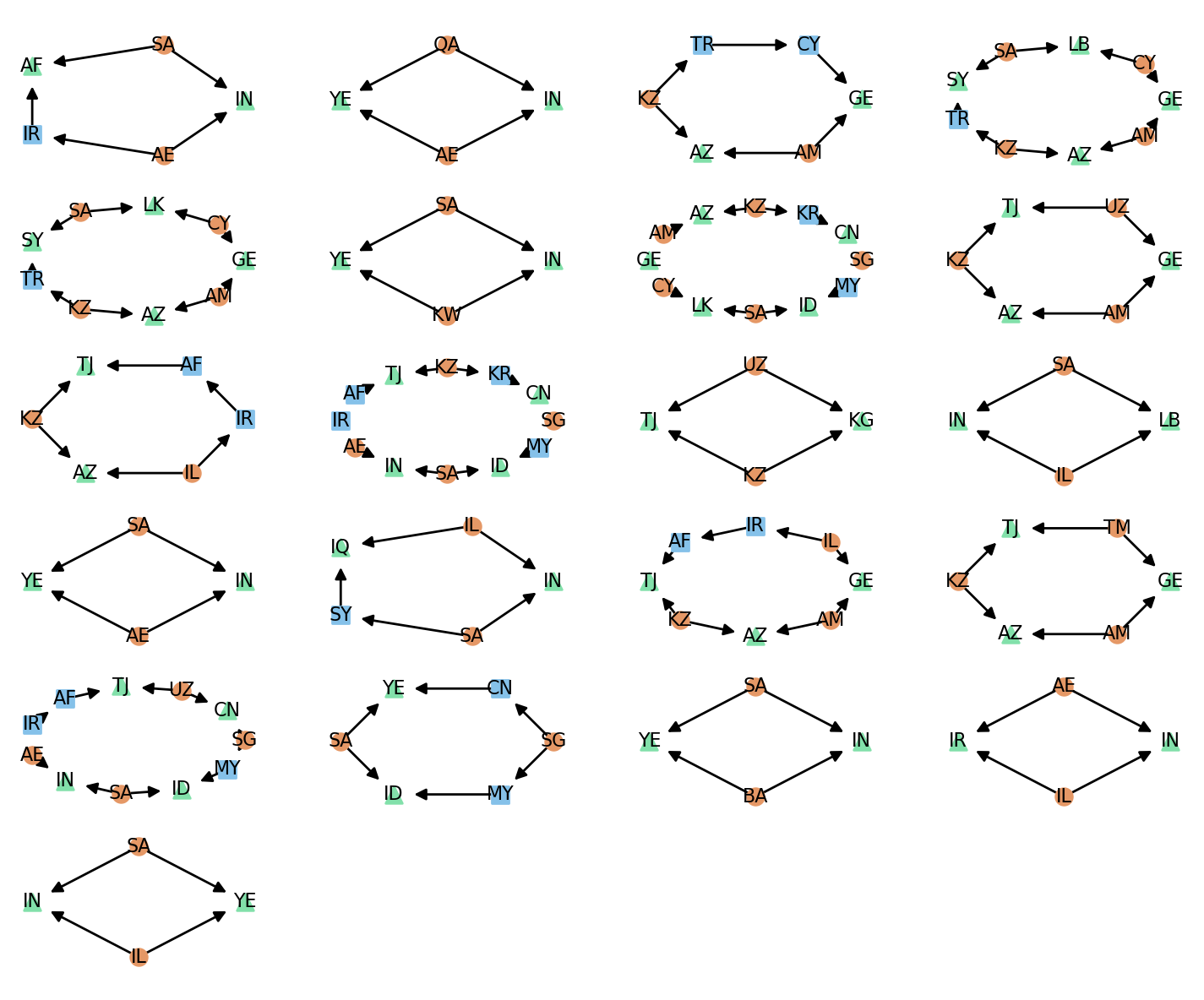}
    \caption{Generating cycles for persistent path homology on remittance network(cont.)}
    \label{fig:remittance2}
\end{figure}

\subsection{C.elegans}
\paragraph{Path homology.} 

We first compute the $1$-dimensional path homology of this network. 
The rank of $1$-dimensional homology is 17, whereas it is reported in~\cite{reimann2017cliques} that the rank of $1$-dimensional directed cliques is 183. This means that boundary quadrangles have a great effect on the synaptic contact structure; there are a number of cycles that can be written as linear combination of boundary triangles and quadrangles, but cannot be represented only by boundary triangles. 

We also compare with the $1$-dimensional path homology on a set of 1000 directed Erdo\" s-Re\'nyi random graphs with 279 vertices and the connection probability 0.028(thus with 2194 expected connections). 
The average first betti number of these ER random graphs is 114.91, while the first betti number in C.elegance is 17. It supports the claim from~\cite{varshney2011structural} that the C.elegance graph is significantly different from Erdo\" s-Re\'nyi random graph with the same number of vertices and similar number of edges.

\paragraph{Minimal homology basis.}
We also computed the minimum homology basis (with respect to lengths of cycles) for the C.elegans network.
Interestingly, all the 17 generating cycles in the computed minimal homology basis are (non-boundary) \emph{quadrangles}. Furthermore, there are no triangles or quadrangles which are directed cycles.
However, when we tested on 1000 Erdo\" s-Re\'nyi random graphs, we found that each of them has directed triangles or quadrangles in its respective minimal homology basis. 
Our results provide some information on four neuron subnetworks for the chemical synapse netowrk for C.elegans (the analysis of four neuron subnetworks is of interest and was previously carried out for the gap-junction network of C.elegans \cite{varshney2011structural}).

\section{Concluding remarks}
A natural question is whether it is possible to have a more efficient algorithm for computing (persistent) path homology of higher dimensions improving the work of \cite{chowdhury2018persistent}. 
Another question is whether we can compute a minimal path homology basis faster improving our current time bound $O(m^\omega n)$. 

\paragraph{Acknowledgement.} The authors thank annonymous reviewers for helpful comments on this paper. This work is in part supported by National Science Foundation under grants CCF-1740761, DMS-1547357, and RI-1815697. 

\bibliographystyle{splncs04}
\bibliography{references}
\appendix

\section{Proof of Theorem~\ref{thm:B1}}\label{apex:bnd}
\begin{proof}
First, by discussions in Section \ref{subsec:bgexamples}, it is easy to see that $\mathsf{Q} \subseteq \mathsf{B}_1 \subseteq \mathsf{Z}_1$. 
We now show that $\mathsf{B}_1 \subseteq \mathsf{Q}$. 
To this end, consider any 2-path (2-chain) $p = \sum a_{uvw} \cdot e_{uvw} \in \Omega_2$, where $a_{uvw} \in \mathbb{F}$ is the coefficient of the elementary 2-path $e_{uvw}$. 
Its boundary is $\partial p = \sum a_{uvw} (e_{uv}- e_{uw} + e_{vw})$ which we will argue to be in the space $\mathsf{Q}$. 
Now set $C = \partial p$.

Specifically, consider an elementary 2-path $e_{uvw}$ in $C$ such that $a_{uvw} \neq 0$. We have that $\partial e_{uvw} = e_{uv} - e_{uw} + e_{vw}$. Note that at this point, we do not yet know whether the allowed 2-path $e_{uvw}$ is also an $\partial$-invariant path yet (i.e, it is not clear whether $\partial e_{uvw} \in \Omega_1$). 
Nevertheless, as $e_{uvw}$ is an allowable $2$-path, edges $(u,v), (v, w) \in E$ (that is, both 1-paths $e_{uv}$ and $e_{vw}$ are allowed). 
As for the 1-path $e_{uw}$, we have three cases:

(i) $u = w$, we have a bi-gon.

(ii) $u\not= w$ and $(u, w)$ is an edge in $G$. In this case, $e_{uv} - e_{uw} + e_{vw}$ is allowed. As it is also a cycle, it follows that $e_{uv} - e_{uw} + e_{vw}\in \Omega_1$, and in fact, it forms a boundary triangle. 

(iii) $u\not =w$ and $(u,w) \notin E$: In this case, $e_{uw}$ cannot exist in $\partial C$ because $C \in \Omega_2$, meaning that $\partial C$ must be an allowed 1-path. 
In other words, $e_{uw}$ has to be cancelled by the boundary of some other elementary 2-path $e_{uv'w}$ with $a_{uv'w} \neq 0$ in $C$. Since $(u,w)\notin E$, the $2$-path $e_{uv'w}$ is not a bi-gon nor a boundary triangle.
That is, there is a 2-chain $e_{uv'w}$ with $a_{uv'w} \neq 0$, whose boundary equals $e_{uv'} - e_{uw} + e_{v'w}$, and $(u,v'), (v'w)\in E$. 
Hence $\partial (e_{uvw} - e_{uv'w}) = e_{uv} + e_{vw} - e_{uv'} - e_{v'w}$ forms a boundary quandrangle.

We now update $C' = C - a_{uvw} e_{uvw}$ for cases (i) and (ii), or update $C' = C - a_{uvw} (e_{uvw} - e_{uv'w})$ for case (iii). 
It is easy to see that $C'$ contains fewer terms with non-zero coefficients than $C$. We repeat the above argument to $C'$ till it becomes the zero. As a result, $C = \partial p$ is decomposed to be a combination of bi-gons, boundary triangles and boundary quadrangles. 
It follows that for any 2-chain $p\in \Omega_2$, we have its boundary $\partial p \in \mathsf{Q}$, which proves $\mathsf{B}_1\subset \mathsf{Q}$.
Putting everything together, we have that $\mathsf{B}_1=\mathsf{Q}$.

\end{proof}

\end{document}